\documentclass[a4paper,11pt]{article}
\usepackage{amsfonts}           
\usepackage{amsthm}             
\usepackage{amsmath}
\usepackage{graphicx}
\usepackage{url}
\usepackage{enumerate} 
\usepackage[margin=1in]{geometry}
\usepackage{hyperref}

\newtheorem{theorem}{Theorem}[section]
\newtheorem{definition}[theorem]{Definition}
\newtheorem{proposition}[theorem]{Proposition}

\newtheorem{lemma}[theorem]{Lemma}

\ifdefined\DEBUG

\newcommand{\arnaud}[1]{(A: \textcolor{blue}{#1})}

\newcommand{\jean}[1]{(J: \textcolor{gray}{#1})}
\else
 \newcommand{\ar}[1]{}
 \newcommand{\arnaud}[1]{}
 \newcommand{\je}[1]{}
 \newcommand{\jean}[1]{}
 \fi

\usepackage{hyperref}
\hypersetup{
    colorlinks=true,
    linkcolor=red,
    filecolor=magenta,      
    urlcolor=magenta,
    pdftitle={Overleaf Example},
    pdfpagemode=FullScreen,
    }

\usepackage{amsthm}
\usepackage{amsmath}
\usepackage{amssymb}
\usepackage{mathrsfs}
\usepackage{graphicx}
\usepackage{wrapfig}
\usepackage{floatrow}
\usepackage{enumitem}

\title{Finding Weakly Simple Closed Quasigeodesics on Polyhedral Spheres\thanks{This research was partially supported by the ANR project
Min-Max (ANR-19-CE40-0014)), the ANR project SoS (ANR-17-CE40-0033) and the Bézout Labex, funded by ANR, reference ANR-10-LABX-58.}}
\author{Jean Chartier\thanks{Univ Paris Est Creteil, CNRS, LAMA, F-94010 Creteil, France, jean.chartier@u-pec.fr} \and Arnaud de Mesmay\thanks{LIGM, CNRS, Univ. Gustave Eiffel, ESIEE Paris, F-77454 Marne-la-Vall\'ee, France, arnaud.de-mesmay@univ-eiffel.fr}}

\begin{document}

\maketitle
\begin{abstract}
A closed quasigeodesic on a convex polyhedron is a closed curve that is locally straight outside of the vertices, where it forms an angle at most $\pi$ on both sides. While the existence of a simple closed quasigeodesic on a convex polyhedron has been proved by Pogorelov in 1949, finding a polynomial-time algorithm to compute such a simple closed quasigeodesic has been repeatedly posed as an open problem. Our first contribution is to propose an extended definition of quasigeodesics in the intrinsic setting of (not necessarily convex) polyhedral spheres, and to prove the existence of a weakly simple closed quasigeodesic in such a setting. Our proof does not proceed via an approximation by smooth surfaces, but relies on an adapation of the disk flow of Hass and Scott to the context of polyhedral surfaces. Our second result is to leverage this existence theorem to provide a finite algorithm to compute a weakly simple closed quasigeodesic on a polyhedral sphere. On a convex polyhedron, our algorithm computes a simple closed quasigeodesic, solving an open problem of Demaine, Hersterberg and Ku. 
  \end{abstract}

\section{Introduction}

A geodesic is a curve on a surface, or more generally in a manifold, which is locally shortest. The study of geodesics on surfaces dates back at least to Poincar\'e~\cite{poincare} and led to a celebrated theorem of Lyusternik and Schnirelmann~\cite{lyusternik} proving that any Riemannian sphere admits at least three distinct simple (i.e., not self-intersecting) closed geodesics (while the initial proof of the theorem was criticized, the result is now well-established, see for example Grayson~\cite{grayson}). This bound is tight, as showcased by ellipsoids.

In this article, we investigate closed geodesics in a polyhedral setting. In such a setting, the following relaxed notion is key: a \emph{quasigeodesic} is a curve such that the angle is at most $\pi$ on both sides at each point of the curve. In 1949, Pogorelov~\cite{pogorelov} proved the existence of three simple (i.e., non self-intersecting) and closed quasigeodesics on any convex polyhedron. The proof is non-constructive and it was asked by Demaine and O'Rourke~\cite[Open Problem 24.24]{demaine} whether one could compute such a closed quasigeodesic in polynomial time. Recent progress on this question was made by Demaine, Hersterberg and Ku~\cite{demaine2} who provided the first algorithm to compute a closed quasigeodesic on a convex polyhedron, and their algorithm runs in pseudo-polynomial time. However, their algorithm is ill-adapted to find closed quasigeodesics which are simple -- this has remained an open problem~\cite[Open Problem~1]{demaine2}. Furthermore, as they note, for this problem, ``even a finite algorithm is not known or obvious'': indeed there is no known upper bound on the combinatorial complexity of a simple closed quasigeodesic (for example the number of times that it intersects each edge), so there is no natural brute-force algorithm. We refer to the extensive introduction of~\cite{demaine2} for a panorama on the difficulties in finding closed quasigeodesics, and to~\cite{orourke} for recent results and questions on quasigeodesics on tetrahedra.

\subparagraph*{Our results.} Our contributions in this article are two-fold.

First, we extend the theorem of Pogorelov to a non-convex and non-embedded setting. Precisely, we work in the abstract setting of compact \emph{polyhedral spheres}, which consist of the following data: (1) a finite collection of Euclidean polygons, and (2) gluing rules between pairs of boundaries of equal length, so that the topological space resulting from the gluings is a topological sphere. A face, edge or vertex of a polyhedral sphere is respectively a polygon, an edge or a vertex of one of the polygons, and a vertex is \emph{convex} (respectively \emph{concave}) if the sum of the angles of the polygons around the vertex is at most $2\pi$, respectively at least $2\pi$. Let us emphasize that such a polyhedral sphere is not a priori embedded in $\mathbb{R}^3$. In particular, edges of the triangles might not be shortest paths. This intrinsic description of non-smooth surfaces appears under various names in the literature, see, e.g., piecewise-linear surfaces~\cite{erickson2013tracing} or intrinsic triangulations~\cite{sharp}, and dates back to at least Alexandrov, who proved~\cite[Chapter~4]{alexandrov} that when all the vertices are convex, such a polyhedral sphere is the metric structure of a unique convex polyhedron in $\mathbb{R}^3$ (see~\cite{kane} for an algorithmic version of this result). In the non-convex case, a celebrated theorem of Burago and Zalgaller~\cite{burago}, shows that one can always find a piecewise-linear isometric embedding of a compact polyhedral sphere into $\mathbb{R}^3$, but it might require a large number of subdivisions and the proof has to our knowledge not been made algorithmic.

Note that by definition, a polyhedral sphere is locally Euclidean at every point that is not a vertex. We propose the following generalization of the definition of quasigeodesics to a polyhedral sphere $S$: a closed quasigeodesic is a closed curve that is locally a straight line around any point that is not a vertex, and that is locally a pair of straight lines around a vertex, forming an angle \emph{at most} $\pi$ on each side if the vertex is convex, and forming an angle \emph{at least} $\pi$ on each side if the vertex is concave. A closed curve $\gamma:\mathbb{S}^1 \rightarrow S$ is \emph{simple} if it is injective, and is \emph{weakly simple} if it is a limit of simple curves (see Section~\ref{S:prelim} for details). 

Our first theorem shows the existence of a weakly simple closed quasigeodesic of controlled length on a polyhedral sphere. We denote by $M$ the $\emph{edge-sum}$ of $S$, which we define as the sum of the lengths of the edges of an iterated barycentric subdivision of a triangulation of $S$.

\begin{theorem}[Existence]\label{thm:existence}
Let $S$ be a polyhedral sphere and denote by $M$ its edge-sum. There exists a weakly simple closed quasigeodesic of length at most $M$.
\end{theorem}

The original proof of Pogorelov in the convex case proceeds by first approximating the polyhedron with smooth surfaces, and then taking the limit of the simple closed geodesics on the smooth surfaces, whose existence is guaranteed by the Lyusternik--Schnirelmann theorem. The proof technique for that latter theorem, originating from the work of Birkhoff~\cite{birkhoff}, goes roughly as follows: we consider \emph{sweep-outs}, i.e., a family of simple closed curves sweeping the polyhedron from one point to another point (see Section~\ref{S:prelim} for a precise definition), and consider the sweep-out where the longest curve has minimal length. Then, by applying a \emph{curve-shortening process}, one can use this optimal sweep-out to find simple closed geodesics. This last step is notoriously perilous~\cite{ballmann,btz,grayson}, hence the tumultuous history of the Lyusternik-Schnirelmann theorem. Our proof proceeds by working directly on the polyhedral sphere and we prove the existence of a weakly simple closed quasigeodesic using a similar technique based on sweep-outs. Our key technical contribution is to rely on a curve-shortening process that is well-adapted to the polyhedral structure of the problem: we adapt the \emph{disk flow} originally designed by Hass and Scott~\cite{shortening} for Riemannian surfaces so as to handle the disks formed by the stars of vertices in a seamless way. We are hopeful that this polyhedral variant of the disk flow could find further applications in the study of quasigeodesics.

Theorem~\ref{thm:existence} provides, in addition to the existence of a weakly simple closed quasigeodesic, a bound on its length. Our second result is to leverage this bound in order to control the combinatorics of the quasigeodesic, which allows us to design a finite algorithm to compute a weakly simple closed quasigeodesic on a polyhedral sphere.

\begin{theorem}[Algorithm]\label{thm:algorithm}
Given a polyhedral sphere $S$, we can compute a weakly simple closed quasigeodesic in time exponential in $n$ and $\lceil M/h \rceil$, where $n$ is the number of vertices of $S$, $M$ is its edge-sum, and $h$ is the smallest altitude over all triangles of some triangulation of $S$. 
\end{theorem}

Note that a bound on the length of a quasigeodesic does not translate directly into a bound on the number of times that it crosses each edge of the polyhedral sphere, as these crossings could happen arbitrarily close to vertices, and thus contribute an arbitrarily small length. Our proof of Theorem~\ref{thm:algorithm} investigates the local geometry of quasigeodesics around vertices to show that this does not happen too much, and that one can indeed bound the multiplicity of each edge. Then, our algorithm guesses the correct combinatorics of the simple closed quasigeodesic and checks in polynomial time that it is realizable.

Our proof techniques for Theorem~\ref{thm:existence} only provide the existence of weakly simple quasigeodesics instead of simple quasigeodesics. We believe this to be a necessary evil in any generalization to the non-convex case, as shortest paths accumulate on concave vertices, making it impossible to define a curve-shortening process in the neighborhood of those which preserves simplicity. However, when all the vertices are convex, the result of Pogorelov does show the existence of a (actually three) simple closed quasigeodesics, where we include as a degenerate simple case a curve connecting twice two vertices of curvature at least $\pi$. Furthermore, his proof also provides an upper bound on the length of this simple quasigeodesic, as we explain an the end of Section~\ref{S:existence}. Since our algorithm behind Theorem~\ref{thm:algorithm} only relies on such an upper bound on the length and on the (weak) simplicity of the sought after curve, we can also use it to compute simple closed quasigeodesics in the convex case. This solves Open Problem 1 of~\cite{demaine2}, but note that we are still a long way off a polynomial-time algorithm.

\section{Preliminaries}\label{S:prelim}

In this article, a \emph{polyhedral sphere} is a finite collection of Euclidean polygons, and gluing rules for boundaries of the same length, so that the space obtained by identifying the boundaries of the polygons via the gluing rules is homeomorphic to a sphere. Such a sphere is naturally endowed with a metric which is locally Euclidean at every point except at the vertices of the polygons, where it might display a \emph{conical singularity}: if the total angle of the polygons glued around that vertex is larger than $2\pi$ (respectively at most $2\pi$), we say that the vertex is \emph{concave} (respectively \emph{convex}), and its \emph{curvature} is the angular defect compared to $2\pi$ (which is thus negative for concave vertices). Given a (not necessarily convex) polyhedron described via the coordinates of its vertices in $\mathbb{R}^3$, one can easily compute the underlying polygons and thus the structure as a Euclidean sphere. The reverse direction of embedding a polyhedral sphere in $\mathbb{R}^3$ is significantly more intricate (see~\cite{kane} for the convex case and~\cite{burago} for the general case), hence our choice of the intrinsic model.

Triangulating each polygon defining a polyhedral sphere yields a \emph{triangulated polyhedral sphere}. Furthermore, by doing up to two barycentric subdivisions in each triangle if necessary, we can assume that there are no loops nor multiple edges in this triangulation. Note that this triangulation and these subdivisions do not change the metric of the sphere, only change the altitudes of the triangles by a constant factor and do not impact quasigeodesicity (see next paragraph). Therefore, for convenience, in this article we will always assume that our polyhedral spheres are triangulated and that they contain neither loops nor multiple edges, and we will denote such a sphere by $S$ from now on. The \emph{edge-sum} of such a sphere $S$ is the sum of the lengths of its edges (the difference with the definition in the introduction follows from the preprocessing that we just explained). A \emph{shelling} of a triangulated sphere $S$ is an order $(T_1,\ldots,T_{\ell})$ on the triangles that $S$ consists of so that for all $i \in [1,\ell-1]$, $\bigcup_{k=1}^{k=i}T_k$ is homeomorphic to a $2$-disk $D^2$. It is well-known that all the triangulated spheres are shellable, for example because, by Steinitz's theorem~\cite[Chapter~4]{ziegler} they form the $2$-skeleton of a polytope, and those are shellable~\cite{bruggesser1971shellable}. Throughout this article, we use the following notations for a polyhedral sphere: its vertices are denoted by $p_1,\dots,p_n$, its edges by $e_1, \dots, e_m$ (or sometimes $e_{ij}$ to emphasize the vertices that it connects to) and its triangles by $T_1,\dots,T_{\ell}$. The order induced by the numbering of the triangles is a shelling order. The \emph{star} of vertex $p_i$, denoted by $\mathcal {C}_i$  is the union of the triangles $T_k$ having $p_i$ for common vertex, identified along the edges adjacent to $p_i$.  It is \emph{convex} (resp. \emph{concave}) if $p_i$ is (but note that the shortest path in $S$ between two points of a convex star is not necessarily contained in that star). We optionally rename the vertices of $P$ to have $p_1\in T_1$ and $p_n\in T_\ell$. Finally, we denote by $M$ the sum of the lengths of the edges of $S$, and by $h$ the smallest altitude of all the triangles in $S$. Note that $h$ is a lower bound on the distance between any two vertices. For $\gamma$ an edge or a curve on $S$, we denote by $L(\gamma)$ its length.

A \emph{closed curve} $c$ on $S$ is a continuous map $c:\mathbb{S}^1 \rightarrow S$. A closed curve is \emph{piecewise-linear} if it is locally straight except at a finite number of points. 

\begin{definition} A closed curve is a \emph{quasigeodesic} if it is locally straight around every point of $S$ that is not a vertex, and around a vertex it forms an angle at most (respectively at least) $\pi$ on both sides if the vertex is \emph{convex} (respectively concave). 
\end{definition}

We emphasize that this definition is non-standard in the non-convex case, where it is sometimes simply forbidden for a quasigeodesic to go through a concave vertex~\cite{demaine}. Note that a quasigeodesic is straight around a vertex with zero curvature. A closed curve is \emph{simple} if it is injective. Throughout this article, all the curves will always be parameterized at constant speed. We endow the space of piecewise-linear curves with the uniform convergence metric, i.e., $d(c_1,c_2)=\max_{t\in \mathbb{S}^1}d(c_1(t),c_2(t))$. A closed curve is \emph{weakly simple} if it is a limit of simple curves for this metric: intuitively a weakly simple curve is a curve with tangencies but no self-crossings. We denote by $\mathcal{P}$ the set of constant closed curves, i.e., closed curves $c$ such that there exists $p \in S$ such that $\forall t\in \mathbb{S}^1,c(t)=p$.

We denote by $\Omega$ the space of rectifiable closed curves of length at most $M$. This space is compact for the uniform convergence metric, as can be shown via the Arzel\`a-Ascoli theorem, the bound on the length and the constant-speed parameterization providing equicontinuity (see for example~\cite[Theorem~2.5.14]{buragoivanov}). We denote by $\Omega^{pl}$ the subspace of $\Omega$ consisting of piecewise-linear and weakly simple closed curves. A \emph{monotone sweep-out} of $S$ is a continuous map $\beta:\mathbb{S}^2\longrightarrow S$, where $\mathbb{S}^2$ is seen as the quotient of the cylinder $[0,1]\times \mathbb{S}^1$ by the relation which identifies the circles $(0, \mathbb{S}^1)$ and $(1,\mathbb{S}^1)$ to two points, and such that :
\begin{itemize}
\item $\beta(0,\cdot)$ and $\beta(1,\cdot)$ belong to $\mathcal{P}$, i.e., are two constant closed curves on $S$,
\item $\beta$ has topological degree one,
\item for $s\in (0,1)$, each \emph{fiber} $\beta(s,\cdot):\mathbb{S}^1\longrightarrow S$ belongs to $\Omega^{pl}$, and
\item the sweep-out is monotone, i.e., if $D_s$ denotes the disk to the left of $\beta(s,\cdot)$, the disks $D_s$ are nested: $D_s \subseteq D_{s'}$ for $s'>s$.
\end{itemize}
The requirement on the topological degree informally means that each point is covered once by the sweep-out ; it is there to prevent trivial sweep-outs (for example constant at a point). It can be replaced by the requirement that the starting and endpoints are distinct. The monotonicity corresponds to the third condition, and typical sweep-outs in the literature do not assume it (see~\cite{chambers2021constructing}), but in this paper we will only use monotone sweep-outs and thus for simplicity we will henceforth drop the word monotone. The \emph{width} of a sweep-out is the length of the longest fiber. We denote by $\mathcal{B}$ the space of sweep-outs.

The algorithm underlying Theorem~\ref{thm:algorithm} has complexity exponential in $\lceil M/h \rceil$, i.e., it depends on the actual values of the lengths of the boundaries of the polygons. Therefore, we do not work on a real RAM model and rely rather on a word RAM model, which is powerful enough to express all the operations that we require: see for example~\cite[Section~2]{demaine2} for a description of the $O(1)$-Expression RAM model which can be encoded in the word RAM model and allows for a restricted notion of real numbers and algebraic operations thereon.

\section{Disk flow and sweep-outs}\label{S:diskflow}

We start by describing a monotone sweep-out of controlled width.

\begin{lemma}\label{lem:sweepout}
Let $S$ be a triangulated polyhedral sphere of edge-sum $M$. There exists a monotone sweep-out of $S$ of width at most $M$.
\end{lemma}

\begin{proof}
As explained in the preliminaries, up to subdividing triangles at most twice we can assume that the triangulated sphere $S$ contains neither loops nor multiple edges and is shellable. The monotone sweepout will be obtained by sweeping each triangle in the shelling order. We first describe families of segments sweeping the triangles: we sweep $T_1$ (resp. $T_\ell$) by segments $\sigma_\mathbb{S}^1$ (resp. $\sigma_s^l$) for $s \in [0,1]$, parallel to the side opposite to $p_0$ (resp. $p_n$). Then, for $i$ from $2$ to $\ell-1$ :
\begin{itemize}

\item{If $T_i$ shares a single side with $\bigcup_{k=1}^{k=i-1}T_k$, we sweep $T_i$ by segments $\sigma_s^i$ parallel to this side.}
\item{If $T_i$ shares two sides with $\bigcup_{k=1}^{k=i-1}T_k$, we sweep $T_i$ by segments $\sigma_s^i$ parallel to the third side.}

\end{itemize}

Now, for each $i\in [2,\ell-1]$  and each segment $\sigma_s^i$ through a certain $T_i$, we continuously associate a loop $\gamma_s^i$ formed by the boundary of the disc $\bigcup_{k=1}^{k=i-1}T_k$, deprived of its intersection with $T_i$ and linked to $\sigma_s^i$ by two portions of the edge of $T_i$ (see Figure~\ref{sweep}). For $i=1$ (resp. $i=\ell$), instead we connect $\sigma_s^i$ to $p_0$ (resp. $p_n$) by two portions of $T_i$'s boundary forming a loop $\gamma_s^i$. The loops $\gamma_s^i$ are weakly simple, pairwise do not cross, and form a continuous family with respect to $s$ (and glue appropriately between $\gamma_1^{i}$ and $\gamma_0^{i+1}$). Furthermore, $\gamma_0^1$ and $\gamma_1^\ell$ are points. Therefore they as a whole form a map $\beta:\mathbb{S}^2 \rightarrow S$ of topological degree one since each generic point is covered by one fiber. By construction, $\beta$ has no fiber $\beta(s,\cdot)$ whose length exceeds $M$. Finally, fibers form the boundaries of nested disks, giving us monotonicity. We have thus constructed a monotone sweepout of $\mathcal{B}$ of width at most $M$.
\end{proof}

\begin{figure}[t]
		\begin{center}
			\includegraphics[scale=0.7]{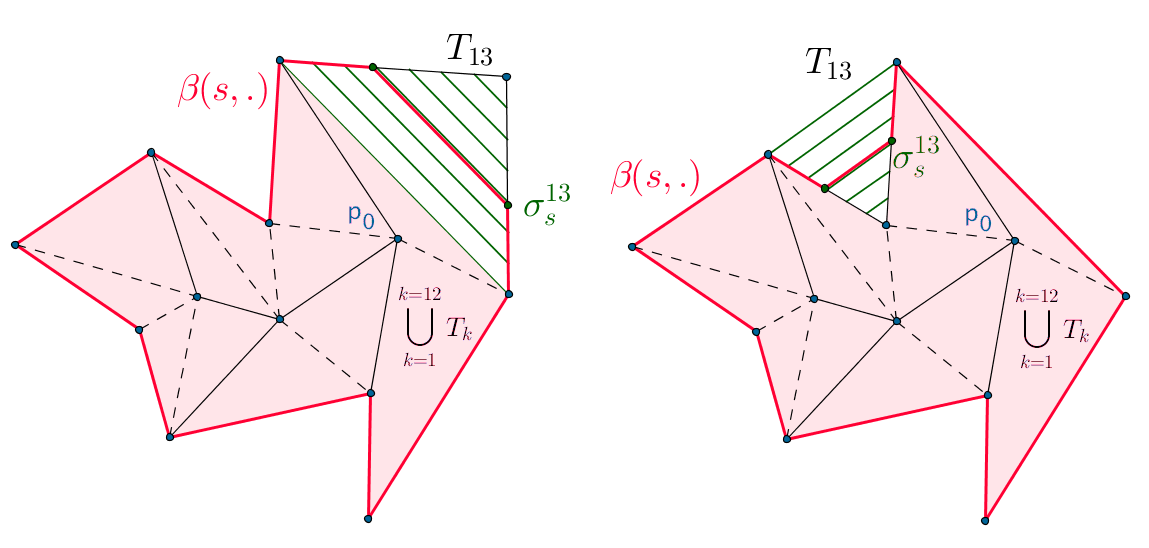}  
			\caption{Construction of a monotone sweepout of $S$}
			\label{sweep}
		\end{center}
\end{figure}

\subparagraph*{The disk flow.} We define here a curve-shortening process that we call the disk flow, which is an iterative process $\Phi$ shortening locally a curve in $\Omega^{pl}$ successively in each star $\mathcal{C}_i$, with the key property that the only fixed points of $\Phi$ are quasigeodesics or trivial curves. In a second step, we will extend $\Phi$ into a map $\hat{\Phi}$ that acts on monotone sweep-outs, which will require interpolating at the points where $\Phi$ is discontinuous. This disk flow is directly inspired by the work of Hass and Scott~\cite{shortening} who defined an analogous flow on Riemannian surfaces. The key difference with their setup is that the star $\mathcal{C}_i$ around a convex vertex is not strongly convex (i.e. there is no uniqueness of shortest paths), which causes additional tears when extending $\Phi$ to sweep-outs and thus requires further operations. Furthermore, instead of working with very small convex disks as they are doing, we work directly with the stars $\mathcal{C}_i$ as we strive to preserve curves whose piecewise-linear structure matches that of $S$. This requires us to deal with tangencies with the boundaries of stars in a different manner.

Let $c$ be a curve in $\Omega^{pl}$ and let $\mathcal{C}_i$ be a star crossed by $c\in \Omega^{pl}$. An \emph{arc} of $\mathcal{C}_i$ is a restriction of $c$ whose image is a connected component of $\mathcal{C}_{i}\cap Im(c)$.
Let $\gamma$ be an arc of $\mathcal{C}_i$, from a closed curve $c$. The points $\gamma(t_0)=c(t_0)\in \partial \mathcal{C}_i$ such that $c([t_0-\varepsilon,t_0))$ or $c((t_0,t_0+\varepsilon])$ is contained in the interior of $\mathcal{C}_i$ for a small enough $\varepsilon >0$, are called the \emph{gates} of $\gamma$. Note that two kinds of arcs have no gates: loops strictly inside the star and arcs never meeting the interior of the star. Unless $\gamma$ is included in $\mathcal{C}_i$, the orientation of $\mathbb{S}^1$ naturally designates a first gate, denoted by \textbf{front}$(\gamma)$, and a final gate, denoted by \textbf{exit}$(\gamma)$. The gates can give access to the interior of the star for values of $t$ greater (resp. less) than $t_0$ -- we say that the gate is open to the right (resp. to the left). A gate can be open to the right and to the left. Thus, front gates are open to the right and exit gates are open to the left. Figure~\ref{gates} illustrates different possible sequences of gates.

\begin{figure}[t]
		\begin{center}
			\includegraphics[width=\textwidth]{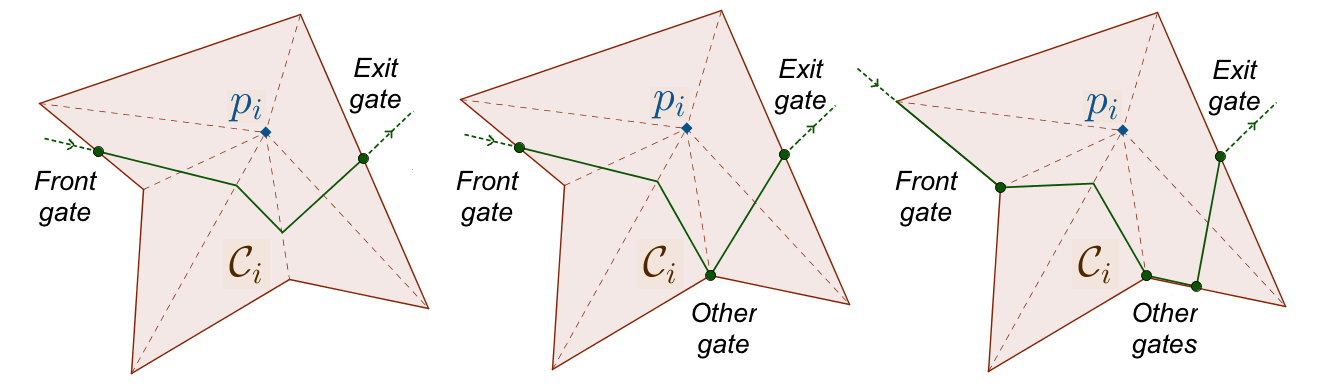} 
			\caption{Three examples of sequences of gates.}
			\label{gates}
		\end{center}
\end{figure}
Relative to two gates $A$ and $B$ \textit{and independently of the path followed between A and B}, we define the \textbf{right region} $\mathcal{C}_i^r(A,B)$ and the \textbf{left region} $\mathcal{C}_i^{\ell}(A,B) $ of the star, as being the two parts of $\mathcal{C}_i$ whose union is $\mathcal{C}_i$ and which intersect along the edges $[Ap_i]$ and $[p_iB]$. The orientation right/left is chosen compatible with that of $c$ between the two gates. The angles of the regions at the $ p_i $ vertex are called the \textbf{right angle} $\theta_r(A,B)$ and the \textbf{left angle} $\theta_{\ell}(A,B)$.

\begin{lemma}
Let $c$ be a curve in $\Omega^{pl}$. The map $\Phi : \Omega^{pl} \longrightarrow \Omega^{pl}$ whose construction we give below verifies the following properties :
\begin{itemize}
\item{The only fixed points of $\Phi$ are quasigeodesics and constant curves.}
\item{$L(\Phi(c))\leq L(c)$, with equality if and only if $c$ is a fixed point.}
\end{itemize}
\end{lemma}

We stress that the map $\Phi$ is in general \emph{not} continuous.\\

  \noindent \textit{Construction and proof.} We define $\Phi$ as follows. Let $c$ be a closed curve in $\Omega^{pl}$. We pick an arbitrary order on the vertices of $S$, which induces an arbitrary order on the stars $\mathcal{C}_i$. The map $\Phi$ consists in repeating in this order a straightening process $\Phi^i_{loc}$ successively in each star. Consider in $\mathcal{C}_i$ an arc $\gamma$ of $c$. Note that between two of its consecutive gates, $A$ open to the right and $B$ open to the left, $\gamma$ lies in $\mathcal{C}_i$.

If $\mathcal{C}_i$ is convex, the straightening is defined as follows for each subset of $\gamma$ between two consecutive gates (which by a slight abuse of notation we also denote by $\gamma$):

\begin{itemize}
\item{If $p_i\in \gamma$ and if $\theta_r(A,B)$ and $\theta_{\ell}(A,B)$ are less than or equal to $\pi$, we replace $\gamma$ by $[Ap_i]\cup[p_iB]$.}
\item{If $p_i\notin \gamma$ and if $\theta_r(A,B)$ and $\theta_{\ell}(A,B)$ are less than or equal to $\pi$, we replace $\gamma$ by the shortest path between $A$ and $B$ staying in the same region relative to $A$ and $B$.}
\item{If $\theta_r(A,B)$ (resp. $\theta_{\ell}(A,B)$) is strictly greater than $\pi$, we replace $\gamma$ by the shortest path between $A$ and $B$ in $\mathcal{C}_i^{\ell}$ (resp. $\mathcal{C}_i^r$).}
\end{itemize}

If $\mathcal{C}_i$ is concave, the straightening is defined as follows:
\begin{itemize}
\item{If $\theta_r(A,B)$ and $\theta_{\ell}(A,B)$ are at least $\pi$, and even if $p_i\notin \gamma$, we replace $\gamma$ by $[Ap_i]\cup[p_iB]$.}
\item{If $\theta_r(A,B)$ (resp. $\theta_{\ell}(A,B)$) is strictly less than $\pi$, we replace $\gamma$ by the shortest path between $A$ and $B$ in $\mathcal{C}_i^{r}$ (resp. $\mathcal{C}_i^{\ell}$).}
\end{itemize}

In case $\gamma=c$ is strictly included in the interior of $\mathcal{C}_i$, then $\Phi^i_{loc}(c)=0$, where $0$ denotes an arbitrary constant curve based at a point $p$ in $\mathcal{C}_i$.

We denote by $\Phi^i_{loc}$, relative to a given star $\mathcal{C}_i$, the straightening process described above, applied in this star to each arc of a closed curve $c \in \Omega^{pl}$. Then $\Phi$ is defined as the concatenation $\Phi:= \circ_{i=1}^n \Phi^i_{loc}$. Let us first show that $\Phi$ has values in $\Omega^{pl}$, note that it suffices to prove it for $\Phi^i_{loc}$. It is immediate that the image under $\Phi^i_{loc}$ is piecewise-linear. In order to prove that the image is weakly simple, we look at the case of two arcs of the same closed curve $c$ in a star, one delimited by two gates $A$ and $B$, the other delimited by two gates $A'$ and $B'$. As $c$ belongs to $\Omega^{pl}$, the two arcs do not cross, so they delimit a band in the star. If $\Phi^i_{loc}$ sends both arcs to the same side of $p_i$, then their images form two shortest paths in the same region and do not intersect. If $\Phi^i_{loc}$ sends the two arcs on opposite sides of $p_i$, a configuration where the two arcs cross twice is impossible because the angles $\theta_r(A,B)$ and $\theta_{r}(A',B')$ on the one hand, and $\theta_{\ell}(A,B)$ and $\theta_{\ell}(A',B')$ on the other hand are arranged in the same order. 

If $c$ is a quasigeodesic, each of its arcs possibly behaves in two ways in the star it crosses: either it reaches and leaves the vertex in a straight line from and up to the boundary of the star, forming on each side an angle at most $\pi$ in the convex case, or at least $\pi$ in the concave case. Or it connects its gates via a shortest path, entirely contained in the more acute of the two regions that it induces. In both cases, the previous process does not change its trajectory. So $\Phi$ fixes the quasigeodesics. Conversely, if $c$ is not a quasigeodesic, then either it does not take a shortest path through a face or in neighborhood of a transverse intersection with an edge, either it forms on the passage of a vertex an angle greater than $\pi $ on one side. This will be straightened when applying $\Phi^i_{loc}$ in a star containing that face, edge, or vertex in its interior, and therefore $c$ is not a fixed point of $\Phi^i_{loc}$. By construction, since $\Phi_{loc}^i$ does not increase lengths, we have that $L(\Phi(c))\leq L(c)$. Let us show that if $ L (\Phi (c)) = L (c) $, then $ c $ is a quasigeodesic. If an arc of $ c $ is not fixed by $ \Phi^i_{loc} $ in a star, while remaining on the same side of the vertex, then it loses length, because there is uniqueness of the shortest path within a (left or right) region of a star. On the other hand, if $ \Phi^i_{loc} $ passes an arc on the other side of the vertex (or pushes it against the vertex), it is because its length exceeds $ L ([Ap_i]) + L ([p_iB ]) $. So the arc loses at least this excess in length. Finally, since some $\Phi^i_{loc}$ decreases the length of a non-quasigeodesic $c$, such a $c$ cannot be a fixed point of $\Phi$.\qed \\

In this proof, we could have taken the simpler choice of always replacing an arc in a star by a shortest path, irrespective of the angle at the vertex. The more delicate choice that is made here is tailored so as to be able to extend $\Phi$ to sweep-outs in Lemma~\ref{lem:hatphi}.

The following property of the map $\Phi^i_{loc}$ will be useful.

\begin{lemma}\label{lem:ppty}
For all $\varepsilon>0$, there exists $\eta>0$ such that for any curve $c \in \Omega^{pl}$ and for any $i$, if $\Phi^i_{loc}(c) \neq 0$ and $L(c)-L(\Phi^i_{loc}(c))<\eta$, then $\Tilde{d}(c,\Phi^i_{loc}(c)):=\max_x d(x\in c,\Phi^i_{loc}(c))<\varepsilon$.
\end{lemma}

\begin{proof}
  
  Let us assume that there exists $\varepsilon>0$ such that for all $n$ in $\mathbb{N}$, there exists $c_n\in \Omega^{pl}$ and $i=i(n)$ such that both $L(c_n)-L(\Phi^i_{loc}(c_n))<1/n$ and $\Tilde{d}(c_n,\Phi^i_{loc}(c_n))>\varepsilon$. Then one of the points of $c_N$ -- that we note $E$ in the following -- is at a distance at least $\varepsilon$ from $\Phi^i_{loc}(c_N)$. Take $N>\sigma(\varepsilon)$, with $\sigma(\varepsilon)=16D^2/\varepsilon^2$, where $D$ is the largest diameter of all the stars. We consider an arc of $c_N$ in $\mathcal{C}_i$ -- which we will still call $c_N$ -- between two gates $A$ et $B$ fixed by $\Phi^i_{loc}$. We distinguish two cases.\\
  
  In the first case, between $ A $ and $ B $, $\Phi^i_{loc}(c_N)$ is a Euclidean straight line in $\mathcal{C}_i$, noted $[AB]$. It is clear that $E\in \mathcal{C}_i$ and that its distance to $[AB]$ is still at least $\varepsilon$. Then the length of $c_N$ is at least the length of the shortest paths between $A$ and $B$ passing by $E$, by staying on the same side of $p_i$ as $c_N$ (without loss of generality, $[AB]\subset \mathcal{C}_i^r(A,B)$). This shortest path, which we denote by $Av^-(c_N)$ can be decomposed into $AE+BE$, although AE and EB are not necessarily Euclidean segments (they can possibly be broken lines going through $p_i$ or other points of $\partial \mathcal{C}_i$). Therefore, the length loss between $c_N$ and $\Phi^i_{loc}(c_N)$ can be lower bounded by the length loss between $Av^-(c_N)$ and $[AB]$, i.e. $AE+BE-AB$. We then reduce this case to the situation where $ABE$ is an Euclidean triangle, as follows. On one hand, if the star is convex, by the rules defining $\Phi^i_{loc}$, two situations occur :
  \begin{itemize}
    \item Either $Av^-(c_N)$ and $[AB]$ are located on the same side of  $p_i$. If they do not form a Euclidean triangle (whose altitude\footnote{In this proof, we call \textit{altitude} the distance from $E$ to $[AB]$. If the triangle $AEB$ has an obtuse angle at its base, then this notion of altitude does not coincide with the usual notion, i.e. the distance from $E$ to $(AB)$.} from $E$ is at least $\varepsilon$), we expand the star flat by cutting it along a ray that does not pass through $Av^-(c_N)$ and replace the folded parts of $Av^-(c_N)$, on either side of $E$, by $[AE]$ or $[BE]$, which shortens it further and allows us to reason in a Euclidean triangle of base $[AB]$ and altitude $a>\varepsilon$ (see Figure \ref{lemme3} below).
 \begin{figure}[ht]
		\begin{center}
			\includegraphics[width=0.7\textwidth]{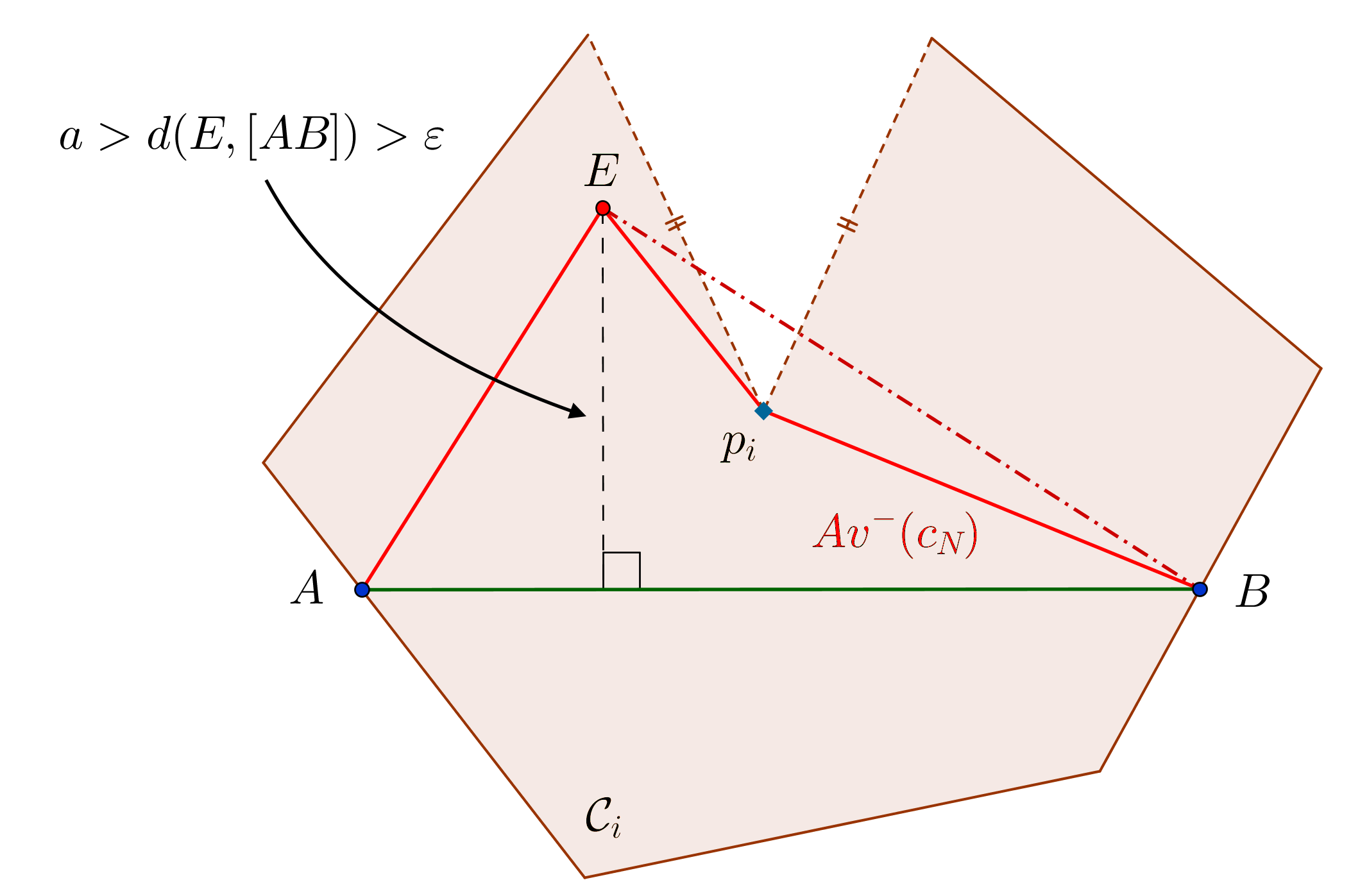} 
			\caption{When $c_N$ and $\Phi^i_{loc}(c_N)$ pass on the same side of $p_i$.}
			\label{lemme3}
		\end{center}
\end{figure}
    
    \item Or $p_i$ is located between $Av^-(c_N)$ and $[AB]$. This only happens if $\theta_l(A,B)$ is greater than $\pi$\footnote{This angle is the key to the reasoning. It prevents $c_N$ and $Av^-(c_N)$ from being arranged as a rhombus, i.e. from remaining distant from each other, with the same length, which should not happen. Indeed, this angle greater than $\pi$ ensures that the curves rather form a "boomerang" between $A$ and $B$.}.  Then we have $AB\leq Ap_i+Bp_i \leq L(Av^-(c_N))$. If $E\in \mathcal{C}_i^r(A,B)$, we replace $Av^-(c_N)$ by $[AE]\cup[BE]$ in this region, which is an even shorter path and brings us back to the previous point. Assume that $E\in \mathcal{C}_i^l(A,B)$ (see Figure \ref{lemme3bis} below). If the distance between $E$ and $Ap_iB$ is greater than $\varepsilon/2$, we draw in $\mathcal{C}_i^l(A,B)$ an avatar $Av^+([AB])$ longer than $[AB]$, to form a triangle of altitude $a_1>\varepsilon/2$ from which the loss of length between $c_N$ and $\Phi^i_{loc}(c_N)$ can be controlled. Otherwise, the distance between $p_i$ and $[AB]$ is necessarily greater than $\varepsilon/2$ and we reason in the same way in the triangle $Ap_iB$ of altitude $a_2>\varepsilon/2$.
    \begin{figure}[ht]
		\begin{center}
			\includegraphics[width=\textwidth]{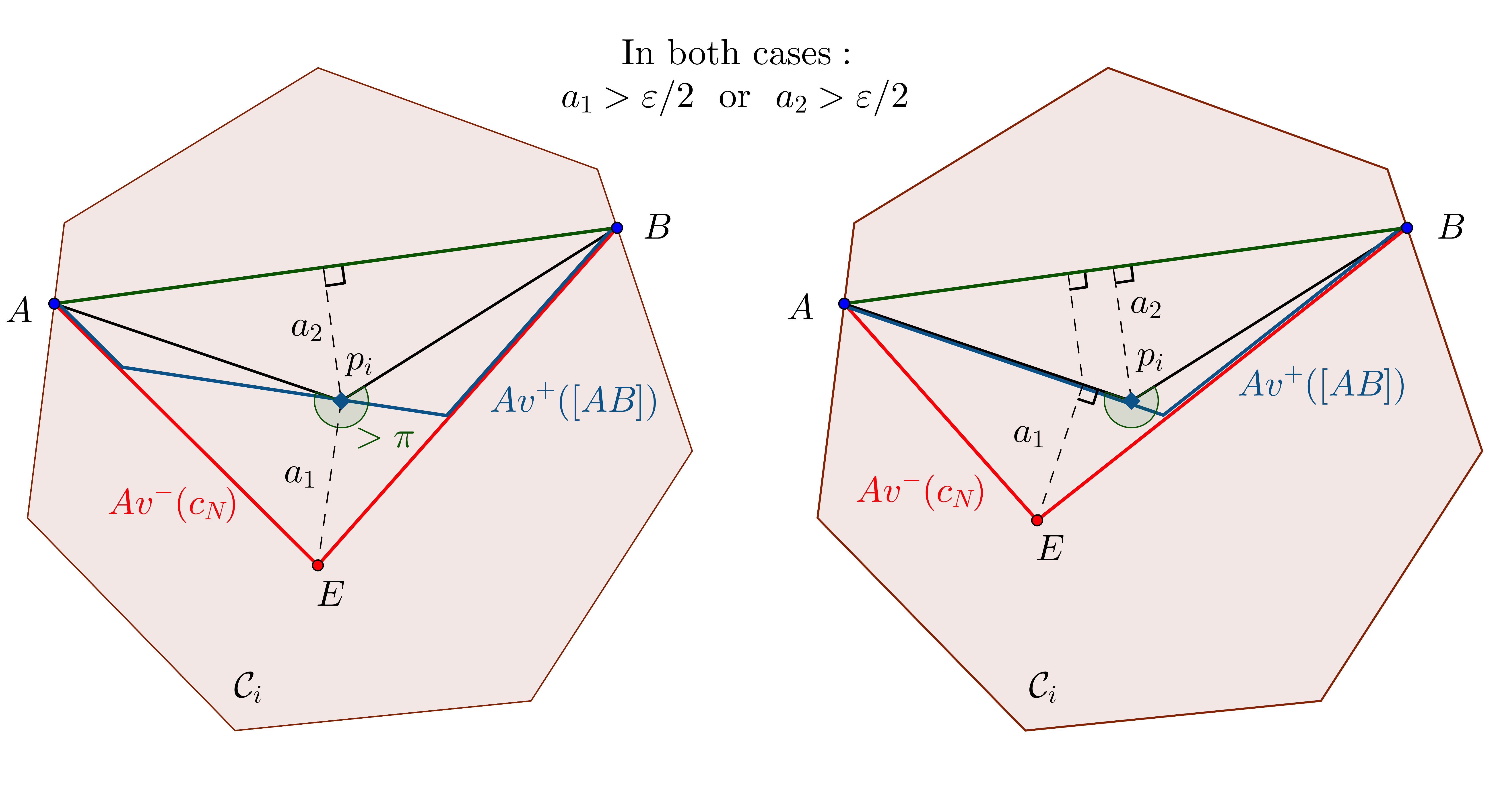} 
			\caption{When $c_N$ and $\Phi^i_{loc}(c_N)$ surround $p_i$.}
			\label{lemme3bis}
		\end{center}
\end{figure}
  \end{itemize} 
  
  \noindent On the other hand, if the star is concave, while the shortest paths $AE$ and $BE$ might be broken lines going through the vertex or the vertex might be inside the triangle, the length loss is greater than it would be in a Euclidean triangle of altitude $a>\varepsilon$ (this is a general fact for metrics of nonpositive curvature, see for example~\cite[Theorem~2.3.3]{jost}). Finally, in the Euclidean case, one can easily check that this length loss is minimized when $E$ is in $E_0$ on Figure~\ref{prop1bis}, where it is at least $\varepsilon^2/16D^2$ for $\varepsilon$ small enough. Therefore, for our choice of $\sigma$ we obtain $L(c_N)-L(\Phi^i_{loc}(c_N))>1/N$, concluding the proof.

\begin{figure}[ht]
		\begin{center}
			\includegraphics[width=0.8\textwidth]{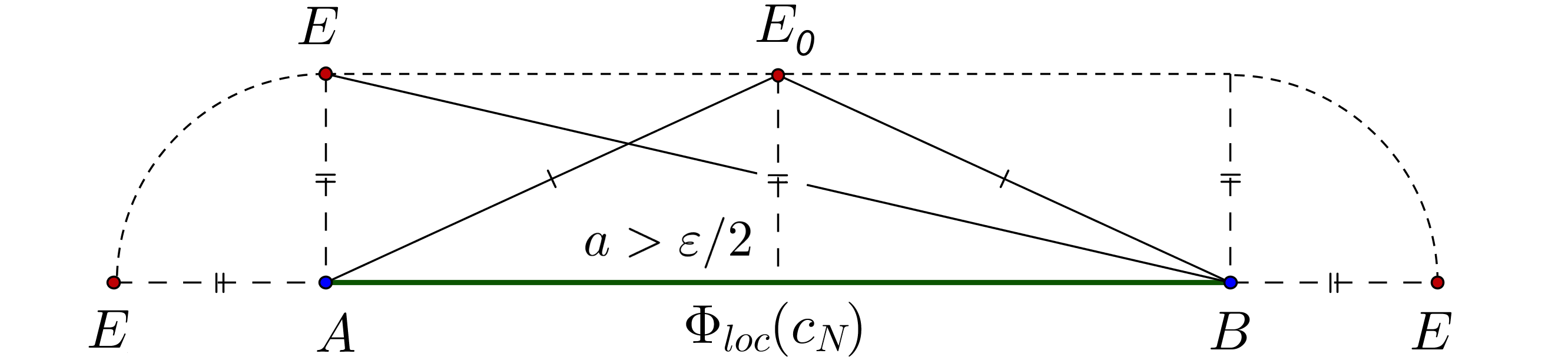} 
			\caption{When controlling a loss of length under the action of $\Phi^i_{loc}$, we control the distance beetwen the curves.}
			\label{prop1bis}
		\end{center}
\end{figure}


In the second case, between $A$ and $B$, $\Phi^i_{loc}(c_n)$ goes along the boundary of $\mathcal{C}_p$, or is a quasigeodesic which passes through $p$. We can then come back to the first case by considering a curve $Av^+(\Phi^i_{loc}(c_n))$ longer than $\Phi^i_{loc}(c_n)$, as pictured in Figure~\ref{property}.
\end{proof}
\begin{figure}[h!]
		\begin{center} \includegraphics[width=0.6 \textwidth]{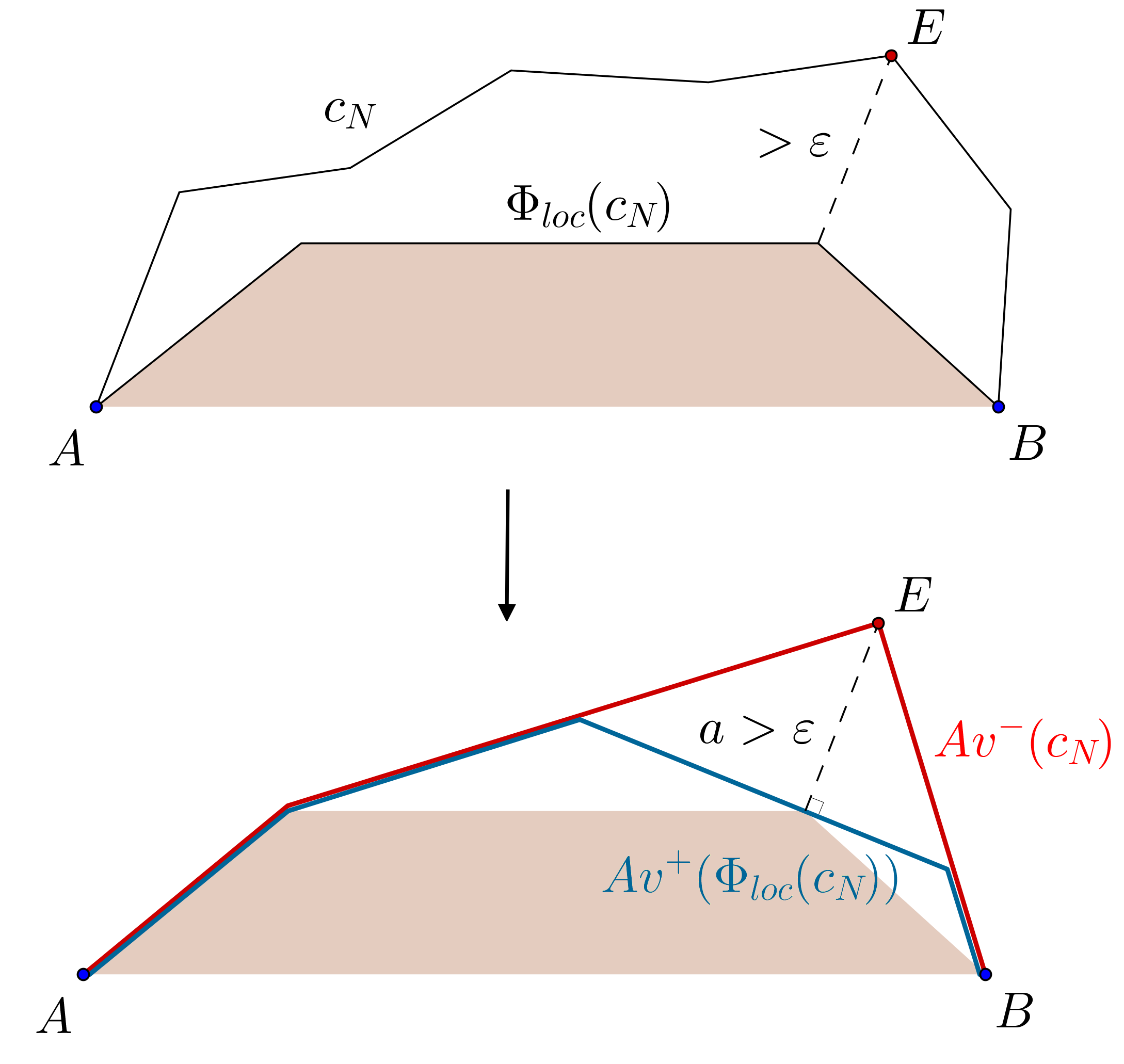}
        \caption{When $\Phi^i_{loc}(c_n)$ goes along the boundary of $\mathcal{C}_p$.}
			\label{property}
		\end{center}
\end{figure}

The following lemma shows that applying $\Phi$ iteratively to a curve either makes the curve trivial in finite time, or converges to a quasigeodesic. Note that the lemma is not as obvious as it might seem as $\Phi$ is not continuous on $\Omega^{pl}$.

\begin{lemma}\label{lem:quasigeod}
Let $c\in \Omega^{pl}$. We consider the sequence of iterates of $\Phi$, i.e., $(\Phi^j(c))_j$. If this sequence does not reach $0$ in finite time, then it admits a subsequence converging to a quasigeodesic (with respect to the uniform convergence metric).
\end{lemma}
\begin{proof}
  Suppose that $\Phi^j(c) \neq 0$ for all $j$. In particular, no curve $\Phi^j(c)$ is strictly contained in a star. Since $\Omega$ is compact, $(\Phi^j(c))_j$ admits a convergent subsequence $(c_k)=(\Phi^{j(k)}(c))_k$ convering to a curve $c_\infty \in \Omega$, which is non-trivial and not contained in a star. Let us assume first that $c_\infty$ is not a quasigeodesic in the neighborhood $V_G \subseteq c_{\infty}$ of a point $G \in S$ contained in the interior of a triangle $T$. This means that three points of $V_G$ are not aligned in $T$, and thus that for $k$ big enough, an arc of $c_k$ also admits three non-aligned points in $T$. However, by construction, the image of a curve of $\Omega^{pl}$ under $\Phi$ is linear in the interior of each triangle of $S$. Thus we reach a contradiction, and thus $c_\infty$ is quasigeodesic in the neighborhood of every point outside of the edges and vertices of $S$. In particular, $c_{\infty} \in \Omega^{pl}$.

Now, let us assume that a non-quasigeodesic point $G$ is contained on an edge $e=(p_{i_1},p_{i_2})$ of $S$, adjacent to two triangles $T_3$ and $T_4$, and without loss of generality we can assume that $G$ is in the interior of $e$ or $G=p_{i_1}$. For $\varepsilon$ to be chosen later, let $\eta=\eta(\varepsilon)$ the length difference given by Lemma~\ref{lem:ppty}. Since the sequence $L(c_k)$ is non-increasing, for $k$ big enough we have $L(c_k)-L(c_{k+1})<\eta$. Furthermore, we also have, for $i=1,\dots,n-1$, $L(\Phi_i(c_k))-L(\Phi_{i+1}(c_k))<\eta$, where $\Phi_\ell$ is the concatenation of the first $\ell$ actions $\Phi^i_{loc}$ on the first $\ell$ stars. By Lemma~\ref{lem:ppty} and the triangle inequality, we thus have $\tilde{d}(c_k,\Phi_{i_1-1}(c_k))<n\varepsilon$. If we replace the connected component of $c_{\infty} \cap (T_3 \cup T_4)$ containing $V_g$ by a shortest path between its endpoints its length decreases by some $\mu >0$. We claim that for $\varepsilon$ smaller than some $\varepsilon_1$ and $k$ big enough, $\Phi_{i_1-1}(c_k)$ is close enough to $c_k$, which itself is close enough to $c_\infty$, so that $\Phi^{i_1}_{loc}$ reduces the length of $\Phi_{i_1-1}(c_k)$ by at least $\mu/2$, i.e., $L(c_k)-L(\Phi_{i_1}(c_k))>\mu/2$. Indeed, if $G$ is in the interior of $e$, a curve close to $c_{\infty}$ stays disjoint from a vertex in $T_3 \cup T_4$, and thus straightening this curve in $\mathcal{C}_{i_1}$ reduces its length by at least the same amount as in $T_3 \cup T_4$. If $G=p_{i_1}$, then any arc close enough to $c_{\infty}$ in $\mathcal{C}_{i_1}$ will have gates inducing a wrong angle at the vertex (since $c_\infty$ does), and thus $\Phi_{loc}$ replaces this arc by a shortest path, away from $p_{i_1}$, and here again the length loss is at least $\mu_2$. Finally, for $\varepsilon$ smaller than some $\varepsilon_2$, we have $\eta < \mu_2$. Taking $\varepsilon < \min(\varepsilon_1,\varepsilon_2)$, we reach a contradiction. We conclude that $c_{\infty}$ is a quasigeodesic.
\end{proof}

We now explain how to apply the disk flow to a monotone sweep-out, so that it extends the action on each of the fibers.

\begin{lemma}\label{lem:hatphi}
The map $\hat{\Phi}:\mathcal{B} \longrightarrow \mathcal{B}$ whose construction we give below is provided with a piecewise continuous injective map $\iota:[0,1]\longrightarrow [0,1]$, such that \[\forall s\in [0,1],\hat{\Phi}(\beta)(\iota(s),\cdot)=\Phi(\beta(s,\cdot)).\] The map $\iota$ induces a surjection $f$ that maps $[0,1]$ on to $[0,1]$, which continually extends $\iota^{-1}$, with the property that $L(\hat{\Phi}(\beta)(s,\cdot))\leq L(\beta(f(s),\cdot))$, with equality if and only if $\beta(f(s),\cdot)$ is a quasigeodesic.
\end{lemma}
  \noindent \textit{Construction and proof. }Let $\beta$ be a sweep-out in $\mathcal{B}$. We explain how to apply a local step $\hat{\Phi}^i_{loc}$ of the curve-shortening process to $\beta$. Then, as before, we will define $\hat{\Phi}$ as the concatenation $\circ_{i=1}^n \hat{\Phi}^i_{loc}$.

Before analyzing the effect of $\Phi^i_{loc}$ on $\beta$, we apply an artificial thickening of $\beta$ which fills its ``problematic’’ portions on the boundary of each star and is defined as follows.
  We call the \emph{bare boundary} of $\mathcal{C}_i$ the set of points of $\partial \mathcal{C}_i$ which are not the gates of any arc of a fiber of $\beta$ crossing $\mathcal{C}_i$.   Consider a connected component of the bare boundary of a certain star $\mathcal{C}_i$. It is fully contained in the image of at least\footnote{If there is an infinite number of them, they are parameterized in $ \beta $ by a closed interval. We then consider the representative closest to the interior of the star.} one fiber $c$ of $\beta$ that:
\begin{itemize}
\item{either connects two gates which are neither a front gate nor an exit gate,}
\item{or it connects a front or exit gate on one side only (see the green curve on Figure~\ref{Bareedge}, top left),}
\item{or it does not connect any gate (see the green curve on Figure~\ref{Bareedge}, top right).}
\end{itemize}

In all three cases, we can see that applying $\Phi^i_{loc}$ would induce a discontinuity around $c$. This is pictured in Figure~\ref{Bareedge}, where one sees that the action of $\Phi^i_{loc}$ on the red curve and the green curve would be very different, despite them being arbitrarily close. We handle this discontinuity as follows. Case 1 will fit into the more general surgery described below, and thus is not addressed at this stage. In cases 2 and 3, the idea is to replace the parameter $s$ of $c=\beta(s,\cdot)$ by a closed interval describing a collection of copies of $c$ all identical (hence the artificial nature of this thickening), except that we drag artificially the position of the single extremal gate (case 2) or we add two new front/exit gates (case 3), one of which moves along $\partial \mathcal{C}_i$. In both cases, the new gates keep or gain an open character to the right or to the left. The aim of this operation is that the arcs of $c$ between these new artificial gates will become straightened by $\Phi^i_{loc}$, thus ensuring the continuity of $\Phi^i_{loc}$ at $c$ (see Figure~\ref{Bareedge}). 

\begin{figure}[h!]
		\begin{center}
			\includegraphics[width=\textwidth]{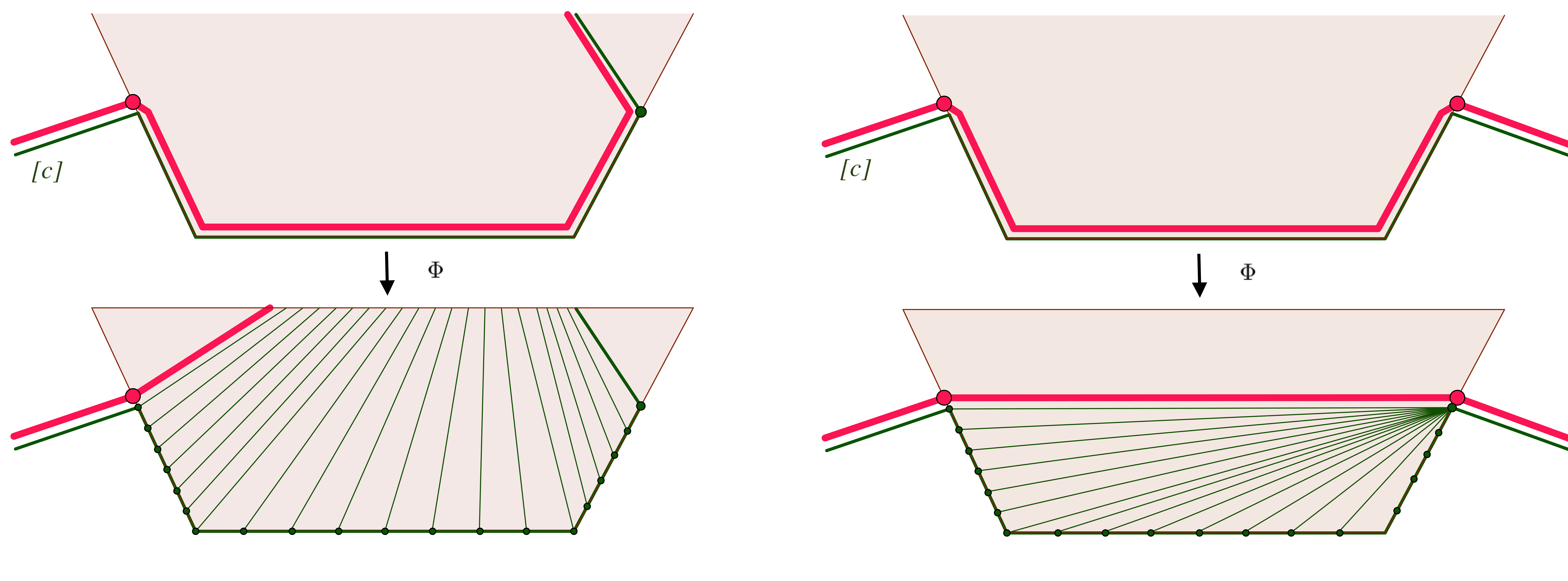} 
			\caption{We artificially add gates on bare edges to obtain interpolating curves in their neighborhoods.}
			\label{Bareedge}
		\end{center}
\end{figure}

After this pre-processing, we consider the map $\beta':[0,1]\times \mathbb{S}^1\longrightarrow S$ defined by: 
\[\forall s\in [0,1],\beta'(s,\cdot)=\Phi^i_{loc}(\beta(s,\cdot)).\]
The discontinuity of $\Phi^i_{loc} $ on arcs within the star $\mathcal{C}_i$ induces a finite number of tears in $\beta'$.  Let us make the exhaustive list of the situations where these tears take place and repair them.

\begin{itemize}
\item{Disappearance of one or more gates far from the vertex: consider a closed interval of fibers, intersecting $\mathcal{C}_p$, such that $ \Phi^i_{loc} $ is sending all the fibers on the same side of $p_i$ and that are parameterized by an interval $I=[s_0,s_0+\varepsilon]\subset \mathbb{S}^1$. We denote by $\gamma_s,s\in I$ the corresponding arcs, relatively to $\mathcal{C}_i$ and we treat the case where $\gamma_{s_0}$ has more gates than all of $\gamma_s,s\in I \setminus \{s_0\}$. In that case, $\Phi_{loc}^i$ might be discontinuous on $s_0$ and we say that we opened a \emph{breach} between two gates of $\gamma_{s_0}$, as pictured in the two examples of Figure~\ref{breach1}. In order to interpolate in this breach, at $\gamma_{s_0}$ we introduce a collection of arcs $\text{gap}(\gamma_{s_0})$ parameterized by $t \in [0,1]$, such that $\text{gap}(\gamma_{s_0})(0)=\Phi_{loc}^i(\gamma_{s_0})$ and which interpolates between $\Phi^i_{loc}(\gamma_{s_0})$ and the subsequent continuous family of arcs $s\in(s_0,s_0+\varepsilon]\longmapsto \Phi^i_{loc}(\gamma_s)$. An arc of $\text{gap}(\gamma_{s_0})(t)$ is defined by taking a shortest path between the gate that opens the breach and a point $p(t)$ on the boundary of the breach and then following the rest of $\gamma_{s_0}$ until the gate opening the breach, as pictured in Figure~\ref{breach1}. Note that thus constructed, all the interpolating arcs in $\text{gap}(\gamma_{s_0})(t)$ have a length strictly smaller than that of $\gamma_{s_0}$.}

\begin{figure}[h!]
		\begin{center}
		  \includegraphics[width=\textwidth]{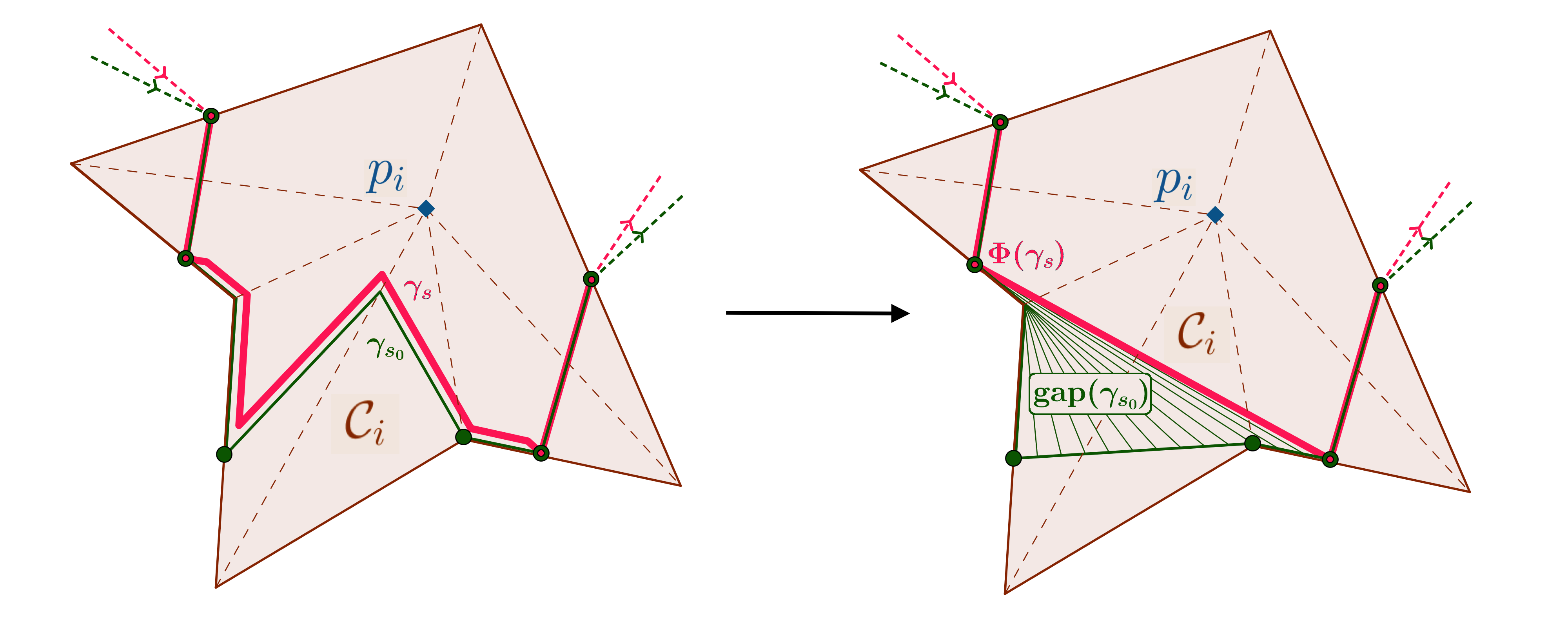}
                  \includegraphics[width=\textwidth]{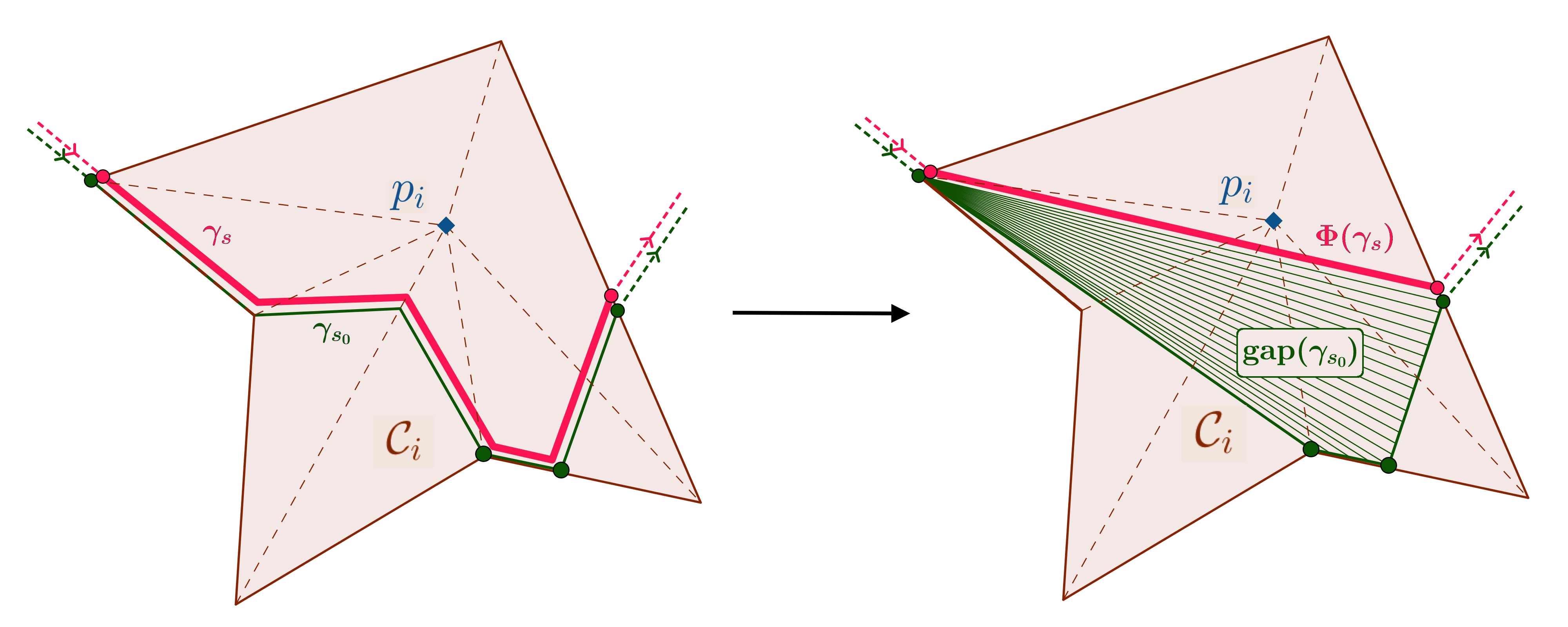} 
			\caption{Disappearance of one or more gates far from the vertex: Two examples of interpolation. In the example in the bottom picture, the $\gamma_{s_0}$ fiber has been added during the preprocessing and provided with an artificial gate. The missing part of the interpolation will be covered by the new gates induced in the preprocessing.}
			\label{breach1}
		\end{center}
\end{figure}

\item Double tear around a convex vertex: Under the action of $\Phi^i_{loc}$, the arcs passing through $p$ which, between two gates, have right and left angles less than or equal to $\pi$ remain attached at $p$. In that case, $\Phi^i_{loc}$ might yield two discontinuities, opening two breaches next to the rightmost arc $\gamma_{r}$ and the leftmost arc $\gamma_{\ell}$, as pictured in Figure~\ref{double}. The two areas to be filled have a triangle as a pattern. Like before, we interpolate into the breach by replacing $\gamma_{r}$ (resp. $\gamma_{\ell}$) with closed arc intervals $\text{gap}(\gamma_{r})$ (resp. $\text{gap}(\gamma_{\ell})$), defined by taking shortest paths to a point $p(t)$ moving continuously on $\gamma_r$ (resp. $\gamma_\ell$) and then following the rest of $\gamma_r$ (resp. $\gamma_\ell$). And as before, the interpolating arcs have length bounded by that of $\gamma_r$ (resp. $\gamma_\ell$). Note that such a breach only happens around convex vertices, since we have uniqueness of shortest paths in a concave star.

\begin{figure}[h!]
		\begin{center}
			\includegraphics[width=.4\textwidth]{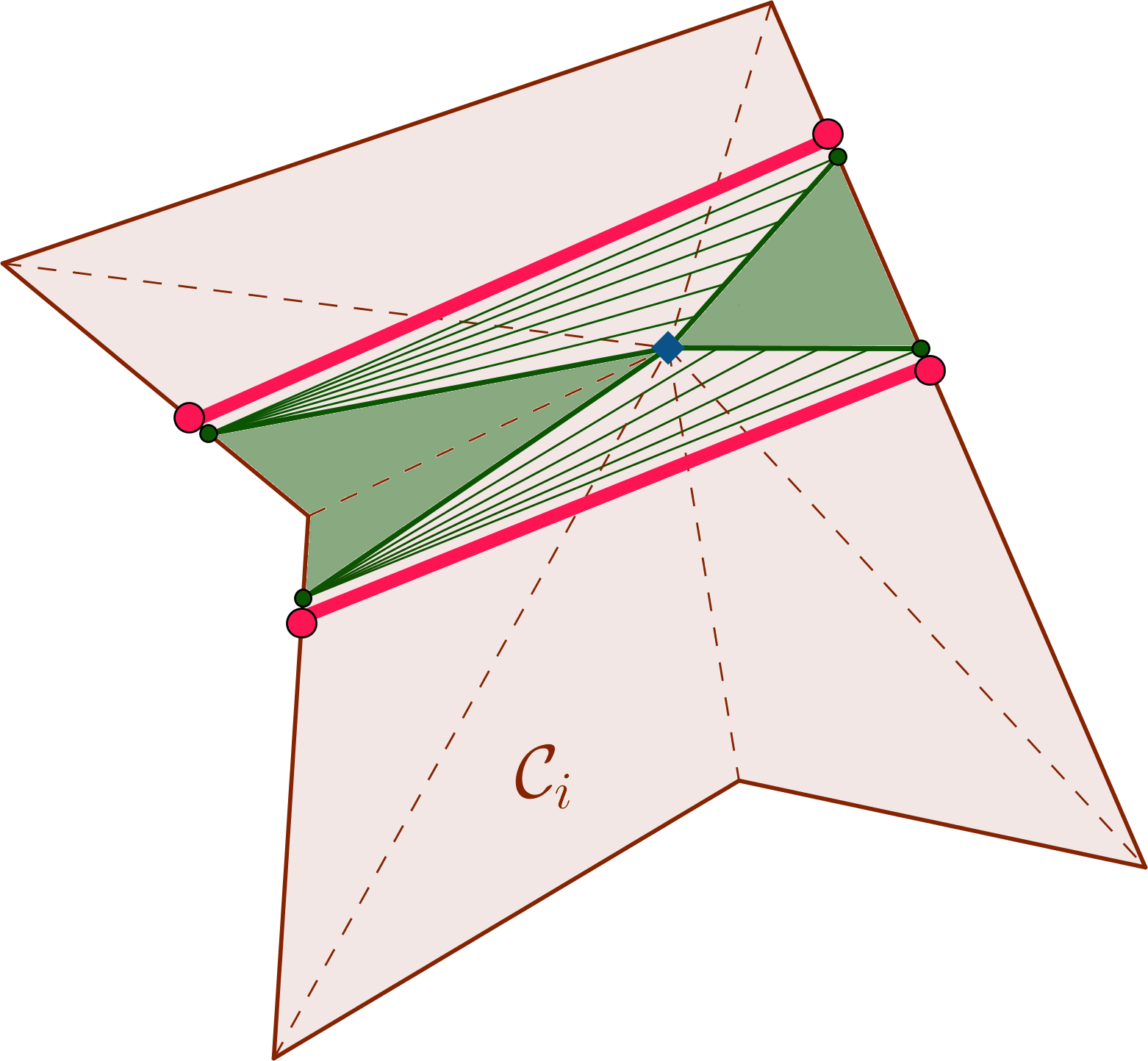} 
			\caption{Double tear around a convex vertex: Interpolating.}
			\label{double}
		\end{center}
\end{figure}

\item Single tear around a vertex: Under the action of $\Phi^i_{loc}$, the arcs passing through $p$ which, between two gates, have one of their angles, right for example, strictly greater than $\pi$ are sent in the opposite region to their greatest angle, the left region to continue the example. To their right, an open interval of arcs is also sent to the left, without creating any discontinuity. At the extremity of this interval we have an arc $\gamma_{r}$ that either defines a right angle equal to $ \pi $ (see Figure~\ref{single}, top), or forms at least one new gate on the boundary of $\mathcal{C}_i$ (see Figure~\ref{single}, bottom). Around this side $\Phi $ is discontinuous and opens a breach on one side of $p$ in the first case, or around $p$ in the second case. The first case is handled exactly as the case of double breaches: we interpolate into the breach by replacing $ \gamma_{r} $ by a closed interval of arcs $\text{gap}(\gamma_{r})$, defined by taking shortest paths to a moving point in $\gamma_r$ and then following $\gamma_r$ (see Figure~\ref{single}, top). The second case is a bit more involved, in some sense it is the combination of the first case and the first item. We first interpolate in the part of the breach that lies in the same region (relative to the gates of $\gamma_r$) of $\mathcal{C}_i$ as $\gamma_r$. This is done by taking a moving point on $\gamma_r$ and taking shortest paths to the moving point \emph{in that region}. Note that the final interpolating arc will pass through the vertex (see Figure~\ref{single}, bottom). Now we can take that arc as if it was an existing fiber, and use it to interpolate the breach in the other region. Here again, this is done by taking a moving point on that arc and taking shortest paths to that arc.

\begin{figure}[h!]
		\begin{center}
		  \includegraphics[width=.7\textwidth]{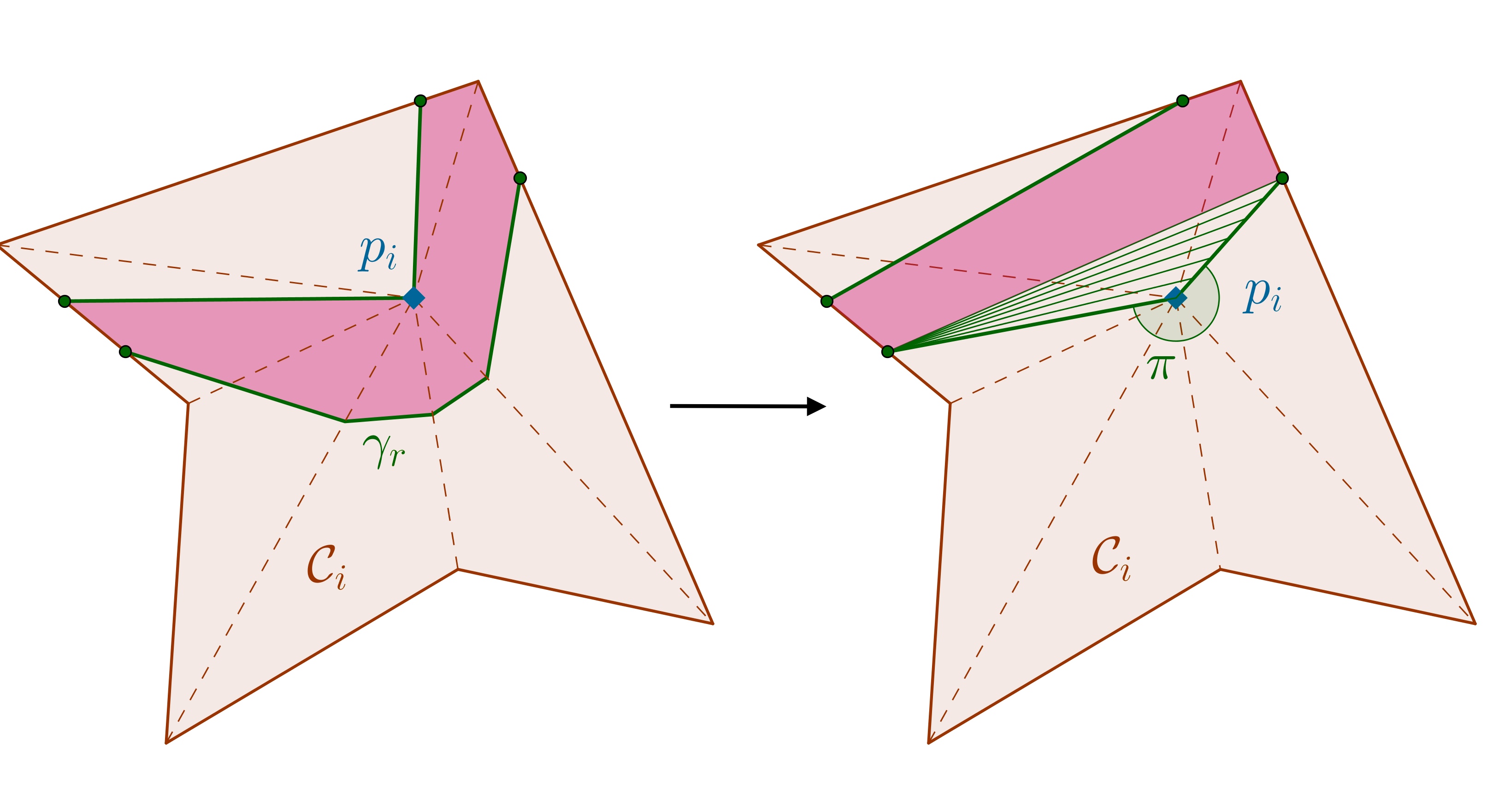}
                  \includegraphics[width=.7\textwidth]{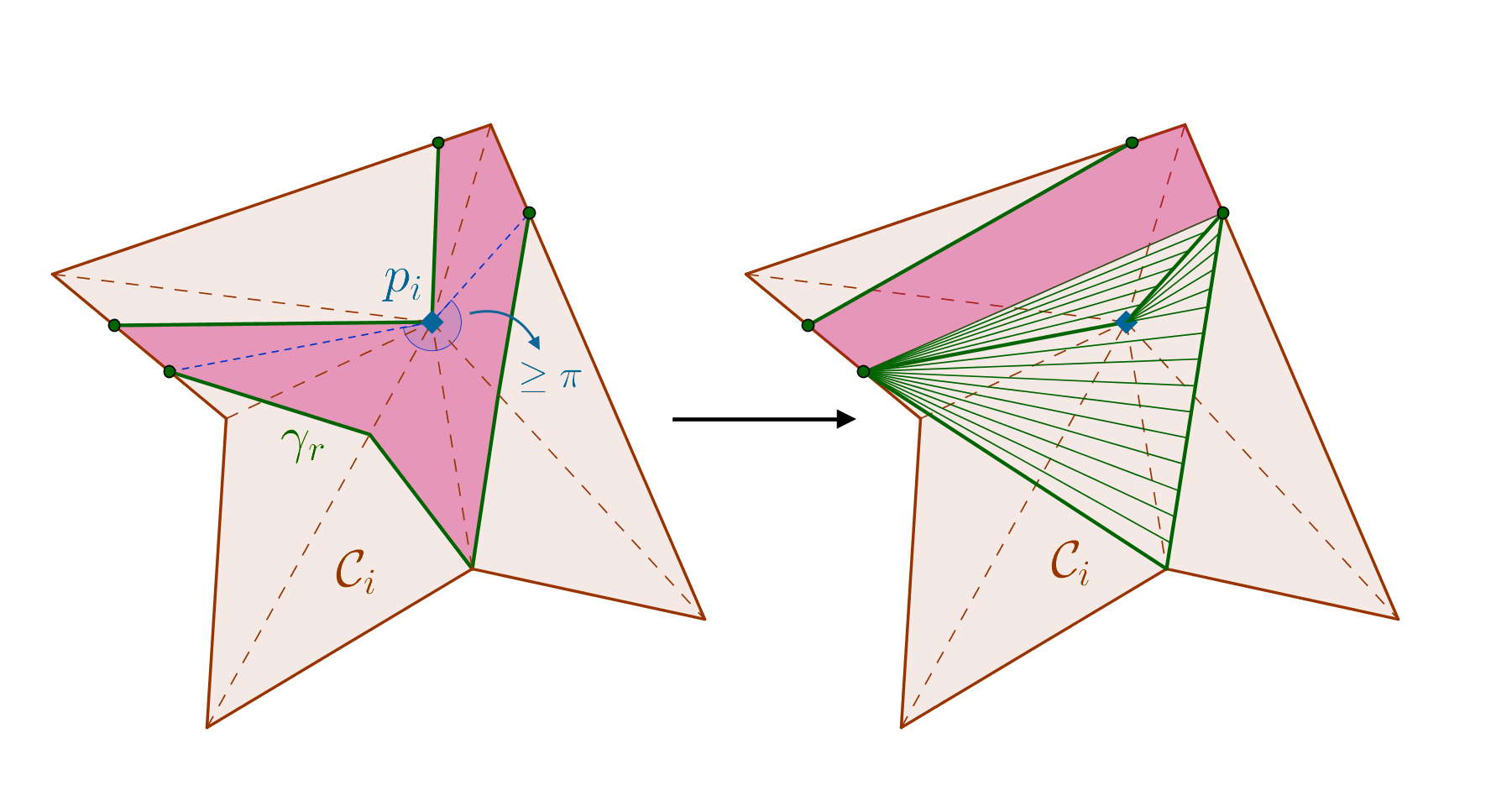} 
			\caption{Single tear around a vertex: Interpolating.}
			\label{single}
		\end{center}
\end{figure}

\item Disappearance of interior curves: A closed curve entirely contained in the interior of a star gets turned into a trivial curve by $\Phi^i_{loc} $. Let $\delta$ denote the greatest parameter such that all the curves parameterized by $[0,\delta)$ are entirely contained in the interior of $\mathcal{C}_i$. Symmetrically, let $\alpha$ denote the smallest parameter such that the curves parameterized by $(\alpha,1]$ are entirely contained in the interior of $\mathcal{C}_i$ (note that $\delta$ and/or $\alpha$ might not exist, then we do nothing on that end). Note that by monotonicity of the sweepout $\beta$, no curve parameterized in $[\delta,\alpha]$ disappears under $\Phi^i_{loc}$. We explain what to do at $\delta$, the situation at $\alpha$ being symmetric. The curve $\gamma_0=\beta(\delta,\cdot)$ coincides with a possibly non-strict subset of $\partial\mathcal{C}_i$. Under the action of $\Phi^i_{loc}$, the curve $\gamma_0$ gets straightened between each pair of gates -- one open to the right, the other to the left. As before, we interpolate within this breach by choosing a point $A$ as a reference point, and then we replace $\gamma_0$ by a closed interval $\text{gap}(\gamma_0)$ of curves, connecting $A$ via shortest path to a point $p(t)$ moving on the boundary of the breach and coming back to $A$ along $\gamma_0$, see Figure~\ref{fig}. Note that there will still be a breach around the vertex $p$ of the star if it is convex, as pictured in the left and right pictures of Figure~\ref{fig}, due to the non-uniqueness of shortest paths between $A$ and some opposite points $B$. This last breach can be filled by continuously moving $A$ and $B$ towards $p$ in such a way that $A$ and $B$ are always connected via a pair of disjoint shortest paths, as pictured in Figure~\ref{gap-trou-final}.

\end{itemize}
\begin{figure}[h!]
		\begin{center}
			\includegraphics[width=\textwidth]{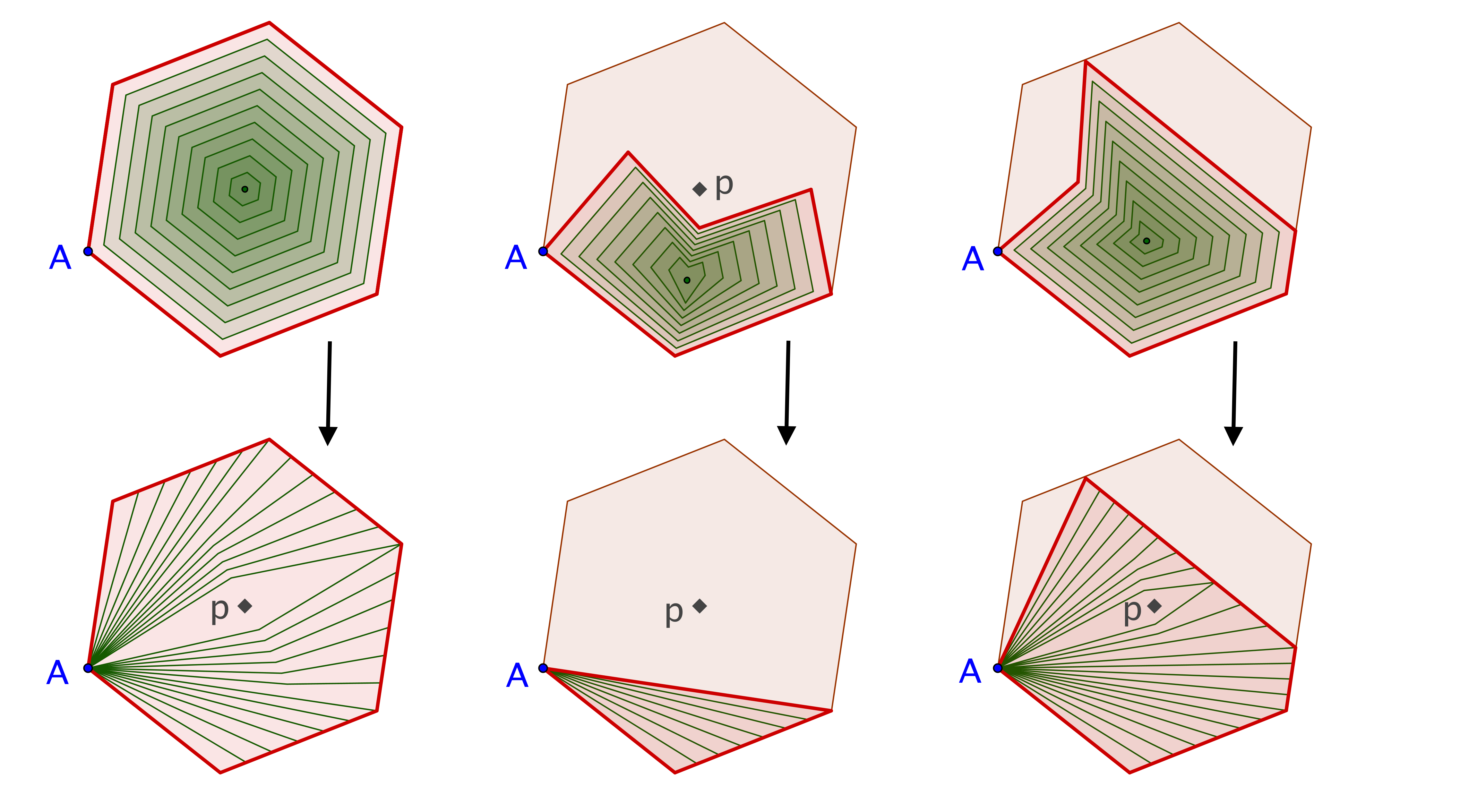} 
			\caption{Interpolating to fill the disappearance of interior curves.}
			\label{fig}
		\end{center}
\end{figure}

\begin{figure}[h!]
		\begin{center}
			\includegraphics[width=0.5\textwidth]{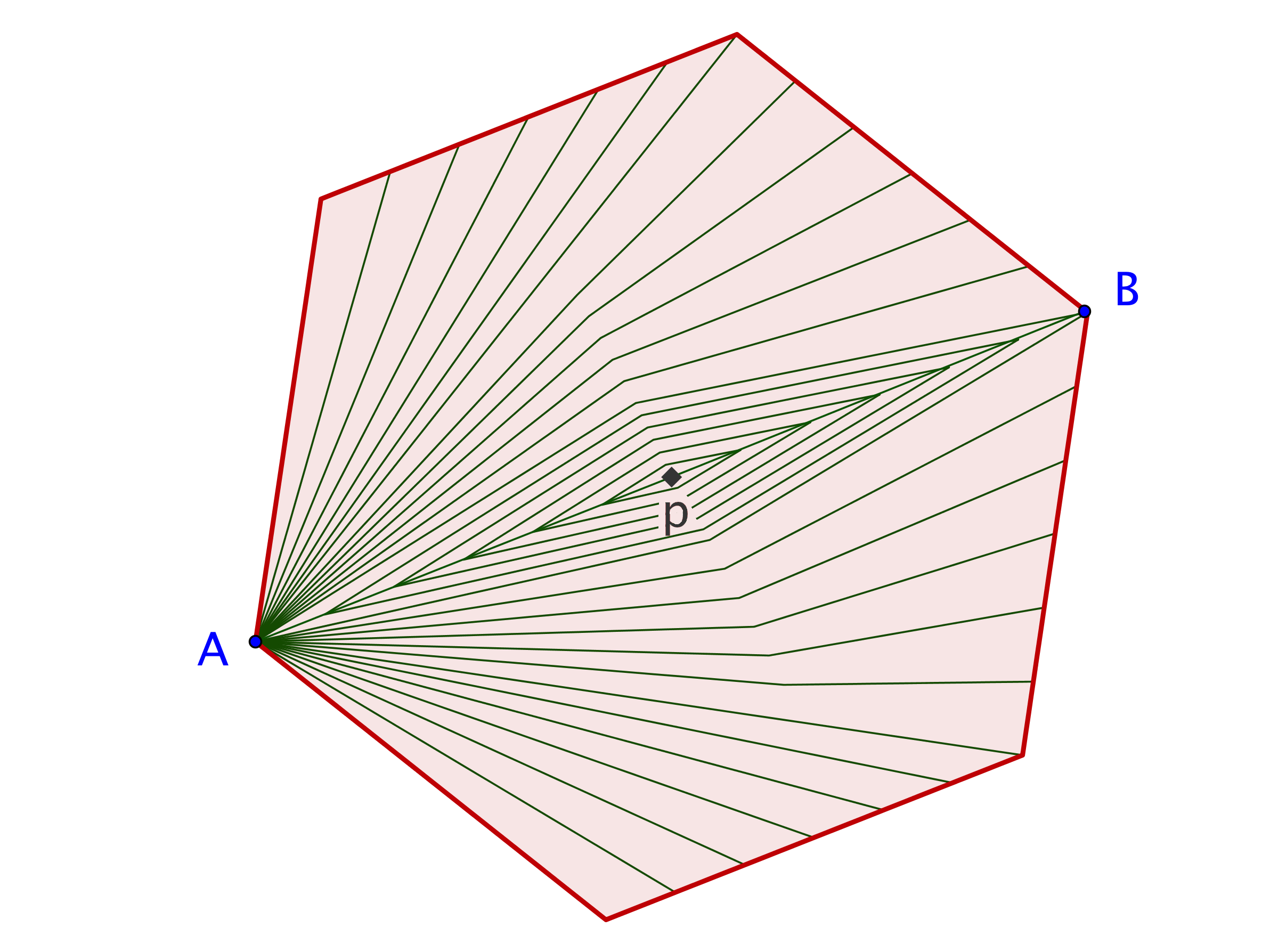} 
			\caption{Disappearance of interior curves: Interpolating.}
			\label{gap-trou-final}
		\end{center}
\end{figure}

In all four cases, the discontinuities have been filled by (1) blowing up the parameter space around a point $s$ to an interval $\text{gap}(s)$ and (2) adding interpolating curves in this interval $\text{gap}(s)$, one of which is $\Phi^i_{loc}(\beta(s,\cdot))$  and all of which have length bounded by that of $\beta(s,\cdot)$, since they are obtained by shortcutting  $\beta(s,\cdot)$ using shortest paths. We define the map $\iota$ as the one sending $s$ to the parameter corresponding to the fiber $\Phi^i_{loc}(\beta(s,\cdot))$, while the surjection $f$ maps the entire interval $\text{gap}(s)$ to $s$ (the maps $i$ and $s$ are defined in the natural way outside of the discontinuities). Therefore we have defined a new map which we denote by $\hat{\Phi}^i_{loc}(\beta)$, whose parameter space is connected to that of $\beta$ using the maps $\iota$ and $f$. As the $\hat{\Phi}^i_{loc}$ get composed to yield $\hat{\Phi}$, the maps $i$ and $f$ are also composed in the natural way.

We argue that the resulting map is a monotone sweep-out. It starts and ends at trivial curves, and by constructions each fiber is piecewise-linear. Furthermore, the disks defined by the fibers are nested, since the effect of $\hat{\Phi}^i_{loc}$ is restricted to the star $\mathcal{C}_i$, where the nesting of disks that was present in $\beta$ is preserved, as the interpolated curves are put inbetween their interpolation targets. Generically, points are covered by the new sweep-out exactly once (since all the fibers can be slightly perturbed to be disjoint), thus the topological degree is one. Finally, since all the interpolating curves have length at most that of a curve it interpolates from, we have the inequality $L(\hat{\Phi}(\beta)(s,\cdot))\leq L(\beta(f(s),\cdot))$, with equality if and only if $\beta(f(s),\cdot)$ is a quasigeodesic.\qed \\

\textbf{Remark:} This proof showcases why our definition of quasigeodesic is the correct one for the disk flow to be appropriately defined on sweep-outs. If we had chosen more strict rules around convex vertices (for example only allowing curves with equal angles on both sides), we could have defined $\Phi$ in a more abrupt way by simply replacing arcs with shortest paths, thus ensuring that no arc through a vertex is fixed by the disk flow. However, this would have yielded tears around a convex vertex $p$ in which our interpolating technique could not have worked, since no fiber of $\beta'$ would be going through the vertex, and there would have been no way to add interpolating fibers of controlled length. In this sense, allowing for an angle at most $\pi$ on both sides is the minimum angular spread allowing for the interpolation steps in the proof of Lemma~\ref{lem:hatphi} to work. For concave vertices, shortest paths between points on the boundary of a star $\mathcal{C}_i$ might require the whole spread of angles at least $\pi$ on both sides, hence this choice of definition.

\section{Existence of a simple closed quasigeodesic}\label{S:existence}

We are now ready to prove Theorem~\ref{thm:existence}. At this stage, our proof follows the same lines as that of Hass and Scott~\cite[Theorem~3.11]{shortening}.

\begin{proof}[Proof of Theorem~\ref{thm:existence}]
  Let $\beta$ be the monotone sweep-out of $\mathcal{B}$ of width at most $M$ described by Lemma~\ref{lem:sweepout}. We consider the sequence of sweep-outs $(\hat{\Phi}^j(\beta))_j$. For any $j \in \mathbb{N}$, the parameter space of $\hat{\Phi}^j(\beta)$ is the product of an interval $[0,1]$ by $\mathbb{S}^1$, the first factor being related to that of $\hat{\Phi}^{j-1}(\beta)$ via the surjection $f_j$ of Lemma~\ref{lem:hatphi}. Therefore, in order to track the history of a fiber in $\hat{\Phi}^j(\beta)$ under the action of $\hat{\Phi}$, we introduce the sequence of parameters $\mathcal{O}_j=(s_0,\dots,s_j)$ such that for all $k$ beetwen $0$ and $j-1$ : $s_{k}=f_{k}(s_{k+1})$. Each space of parameters describing $\mathcal{O}_j$ is homeomorphic to the interval $[0,1]$ (via the trivial homeomorphism $(s_0, \ldots ,s_j) \mapsto s_j$), and we consider the projective limit $\mathcal{I}$ of these intervals, which is thus also homeomorphic to an interval $[0,1]$. An element of this projective limit therefore consists of an infinite sequence $\mathcal{O}=(s_0, s_1 \ldots)$ such that for all $k$, $s_{k}=f_{k}(s_{k+1})$.

  Let $\mathcal{O}=(s_1, s_2 \ldots)$ be an element of $\mathcal{I}$, which thus corresponds to a family of curves $c_j:=\hat{\Phi}^j(\beta)(s_j)$, and let us assume that all these curves are trivial for $j$ bigger than some $k$. Then there is an open neighborhood of $\mathcal{O}$ for which this is also the case, as a curve becomes trivial under the action of some $\hat{\Phi}^i_{loc}$ if and only if it is fully contained in the interior of a star. Therefore, the set $\mathcal{V}^k \subset \mathcal{I}$ of sequences of curves for which the $k$th curve is not trivial is a closed subset of $\mathcal{I}$. Furthermore, it is not empty, as otherwise some intermediate sweep-out after $\hat{\Phi}^{k-1}(\beta)$ would consist of only curves contained in the interior of some star and thus would miss some point of the sphere $S$, in contradiction with the requirement that a sweep-out be of topological degree one. Finally, we have the natural inclusion $\mathcal{V}^{k+1} \subset \mathcal{V}^k$ since if $\hat{\Phi}^{k+1}(\beta)(s_{k+1})$ is not trivial, then this is also the case for $\hat{\Phi}^{k}(\beta)(s_{k})$. We can thus consider the intersection $\cap_{k\in \mathbb{N}} \mathcal{V}^k$ which is an infinite intersection of nested closed non-empty subsets of $\mathcal{I}$ and is thus non-empty. An element in this intersection is a sequence $\mathcal{O_\infty}=(s_1, s_2 \ldots)$ such that none of the curves $c_n=\hat{\Phi}^n(\beta)(s_n)$ is trivial. As $\Omega^{pl}$ is compact, we can extract from this sequence of curves a convergent subsequence $c_k$, which converges to a curve $c_{\infty}$. We claim that $c_{\infty}$ is a weakly simple closed quasigeodesic of length at most $M$. The fact that $c_{\infty}$ is weakly simple follows from the fact that it is a limit of weakly simple curves. The bound on the length follows from the fact that by Lemma~\ref{lem:sweepout}, the width of each of the sweep-outs $\hat{\Phi}^n(\beta)$ is at most $M$, and thus in particular $c_{\infty}$ is a limit of curves of length at most $M$ and thus has length at most $M$, since the length is a lower semi-continuous function on $\Omega$.

  Finally, in order to prove that $c_{\infty}$ is a quasigeodesic, we first introduce the following generalization of Lemma~\ref{lem:ppty}.

  \begin{lemma}[Generalization of Lemma~\ref{lem:ppty}]\label{lem:genppty}
  For all $\varepsilon>0$, there exists $\eta>0$ such that for any sweepout $\beta \in \mathcal{B}$, for any $s\in [0,1]$ and for any $i\in \mathbb{N}$, if $\hat{\Phi}^i_{loc}(\beta(s,\cdot))\neq 0$ and $L(\beta(s,\cdot))-L(\hat{\Phi}^i_{loc}(\beta)(y,\cdot))<\eta$ with $y\in f^{-1}(s)$, then $\tilde{d}(\beta(c,\cdot),\hat{\Phi}^i_{loc}(\beta)(y,\cdot))<\varepsilon$, where $\tilde{d}(c_1,c_2)=\max_{x \in c_1}d(x,c_2)$.
  \end{lemma}
  \begin{proof}
Unlike Lemma \ref{lem:ppty}, it may be that $\beta(s,\cdot)$ is a discontinuity point of $\Phi$ and in particular, $c_y:=\hat{\Phi}^i_{loc}(\beta)(y,\cdot)$ is an interpolation curve, more or less distant from $\Phi_{loc}^i(\beta(s,\cdot))$. Using the notations used in the proof of Lemma \ref{lem:ppty}, $E$ is a point of $c_N:=\beta(s,\cdot)$ at least at a distance $\varepsilon$ from $c_y$. Recall that $c_N$ generates, under the action of $\hat{\Phi}_{loc}^i$, an interval of curves noted gap$(c_N)$. These curves connect two gates $A$ and $C$. The non-continuity of $c_N$ is caused by the fact that $c_N$ passes through at least one point $F$ located either on $\partial \mathcal{C}_i$ or in $p_i$. The relative positions of $A$, $C$ and $F$ give rise to numerous interpolations, essentially described in the Lemma \ref{lem:hatphi}. In particular, the area covered by gap$(c_N)$ may or may not contain $p_i$. What is important in the following reasoning is whether $c_N$ and $c_y$ stay on the same sides or on opposite sides of $p_i$.  We immediately reduce the second case to the first by using the dichotomy argument used in the proof of Lemma \ref{lem:ppty} (Figure \ref{lemme3bis}). Thus we can assume that $c_N$ and $c_y$ are on the same side of $p_i$. Moreover, by the same arguments as in Lemma \ref{lem:ppty} (Figure \ref{lemme3}), we can assume that we are working away from $p_i$, i.e. in a star-shaped portion of the Euclidean plane.

The portion of $c_y$ located between $A$ and $C$ is the concatenation of two shortest paths: the one between $A$ and any point $B$ along the portion of $\Phi(c_N)$ that joins $F$ and $C$ -- the one between this point $B$ and $C$. By hypothesis, $c_N$ also passes, between $A$ and $C$, through the point $E$ which is at least $\varepsilon$ away from $c_y$, i.e. at least $\varepsilon$ away from any segment constituting $c_y$ (recall that $c_y$ is piecewise-linear). We distinguish two cases: when $E$ is located between $A$ and $F$ and when $E$ is located between $F$ and $C$.
\begin{figure}[ht]
		\begin{center}
			\includegraphics[width=1\textwidth]{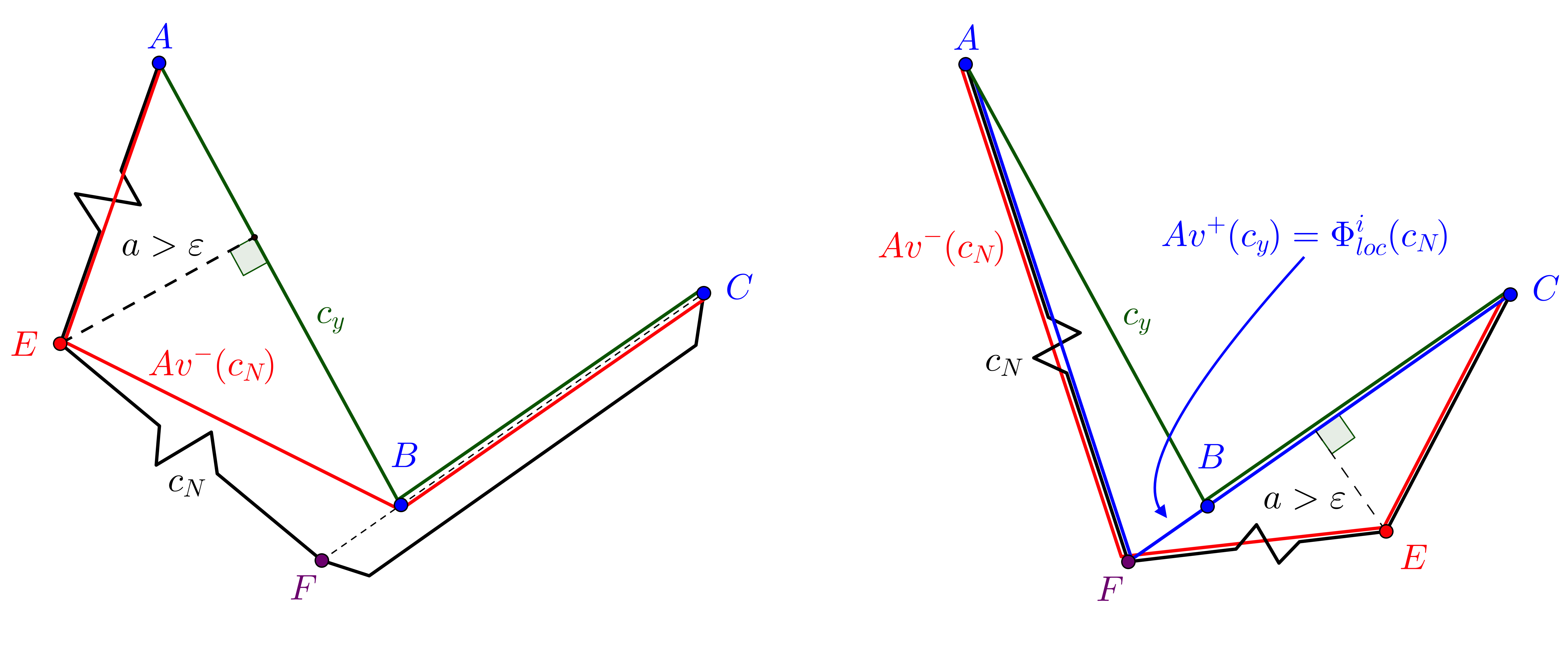} 
			\caption{Generalization of the Lemma \ref{lem:ppty}.}
			\label{lemmegeneral}
		\end{center}
\end{figure}
\begin{itemize}
    \item Case 1, Figure \ref{lemmegeneral}, left : The length loss between $c_N$ and $c_y$ is at least the length loss between a curve $Av^-(c_N)$ -- shorter than $c_N$ -- and the curve $c_y$, with which it coincides, except between $A$ and $B$, where it deviates by $a>\varepsilon$. This situation is then the one of a Euclidean triangle which was handled in the proof of Lemma \ref{lem:ppty}.
    \item Case 2, Figure \ref{lemmegeneral}, right : The length loss between $c_N$ and $c_y$ is at least the length loss between a curve $Av^-(c_N)$ -- shorter than $c_N$ -- and the curve $Av^+(c_y)=\Phi^i_{loc}(c_N)$ -- longer than $c_y$ -- with which it coincides, except between $F$ and $C$, where it deviates from it by $a>\varepsilon$. Here again, the situation is the one of an Euclidean triangle which was handled in the proof of Lemma~\ref{lem:ppty}. Note that if the distance from $E$ to $c_y$ is not realized at the perpendicular of a segment (at $B$ for example), then we readjust $Av^+(c_y)$ in the manner of Figure \ref{property}.
\end{itemize}
 \renewcommand{\qedsymbol}{$\lrcorner$}
  \end{proof}

  Now, the argument is identical to the one in the proof of Lemma~\ref{lem:quasigeod}, to which we refer. If $c_\infty$ is not a quasigeodesic, there is one point $p$ in its image which is not locally quasigeodesic, i.e., there are two points $p_1$ and $p_2$ in a small neighborhood in $c_{\infty}$ such that $p_1$, $p$ and $p_2$ are not aligned, and if $p$ is a vertex, the angle at $p$ is disallowed by the curvature there. For $k$ big enough, $c_k$ will also have this property. Now, we consider a star $\mathcal{C}_i$ which contains $p$ in its interior. By Lemma~\ref{lem:genppty}, $c_k$ will have moved very little when $\Phi^i_{loc}$ acts on it, and thus this action will diminish its length by a fixed quantity that can be lower bounded based on $c_{\infty}$, which is impossible since the lengths of the $c_k$ converge.
\end{proof}

Our techniques only guarantee the existence of a weakly simple closed quasigeodesic of length at most $M$. In contrast, in the convex case, Pogorelov~\cite{pogorelov} proved the existence of a simple closed quasigeodesic (where the degenerate case of two vertices of curvature at least $\pi$ connected twice by a curve is allowed). The proof of Pogorelov works by approximating a convex polyhedron by smooth surfaces, appealing to the Lyusternik-Schnirrelmann on the smooth surfaces to find simple closed geodesics, taking the limit of such simple closed geodesics and arguing that (1) the limit is a quasigeodesic and (2) it is simple. We argue that the same technique proves the existence of a simple closed quasigeodesic of length at most  $M+\varepsilon$, for an arbitrarily small $\varepsilon>0$. Indeed, the sweep-out that we describe on $S$ naturally induces sweep-outs of width at most $M+\varepsilon$ on the approximating smooth surfaces that are close enough, and thus the first simple closed geodesic output by the Lyusternik-Schnirrelmann theorem in each of these surfaces has length at most $M+\varepsilon$. Taking the limit of those yields a simple closed quasigeodesic of length at most $M+\varepsilon$. We will use this result at the end of the next section.

\section{An algorithm to compute a weakly simple closed quasigeodesic}\label{S:algorithm}

In this section, we leverage the existence of a weakly simple closed quasigeodesic of length at most $M$ proved in Theorem~\ref{thm:existence} in order to design an algorithm to find it.

Let $S$ be a polyhedral sphere and denote by $\mathcal{E}=\{p_1,\dots,p_n,e_1,\dots,e_m\}$ be the set of vertices and open edges of $S$. To a closed curve $c:\mathbb{S}^1\longrightarrow S$, we associate the cyclic word $\mathcal{E}(c)$ whose successive letters are the elements of $\mathcal{E}$ met by $c(t)$ as $t$ moves around $\mathbb{S}^1$ (note that an edge can be either crossed or followed). Given a bound on the length of $c$, we want to derive a bound on the combinatorics of $c$, i.e., a bound on the length of $\mathcal{E}(c)$. This is hopeless without any assumption, as a curve spiraling around a vertex  for an arbitrarily long time showcases. But when $c$ is a weakly simple quasigeodesic, we can obtain such a bound. Indeed, our first observation is that a weakly simple quasigeodesic never spirals around a vertex.

\begin{lemma}\label{lem:nonspiral}
Let $\gamma$ be a weakly simple closed quasigeodesic and $\mathcal{C}_i$ be the open star of a vertex $v_i$ of degree $d_i$. Then for any connected component $\alpha$ of $\gamma \cap \mathcal{C}_i$, the number of intersections of $\alpha$ with edges and vertices of $\mathcal{C}_i$ is at most $d_i$.
\end{lemma}

\begin{proof}
  If $\alpha$ passes through the vertex $v_i$, then it exits on both sides tracing a straight line in one of the triangles of $\mathcal{C}_i$. This straight-line reaches directly the opposite edge of the triangle, therefore in this case the number of intersections of $\alpha$ with edges and vertices of $\mathcal{C}_i$ is at most two.

  If $\alpha$ does not pass through the vertex $v_i$, then let us denote by $e$ the first edge adjacent to $v_i$ that it crosses. Note that within a triangle of $\mathcal{C}_i$, by quasigeodesicity, $\alpha$ enters from one edge and does not backtrack, i.e., it escapes from another edge. Therefore, either $\alpha$ escapes from $\mathcal{C}_i$ before crossing $e$ again, in this case it crosses at most $d_i$ edges, or it crosses $e$ again. In the latter case, up to reversing orientation of $\alpha$ we can assume that the second crossing point is closer to $v_i$ than the first crossing point. Tracing $\alpha$ after the second crossing point, we see that in each triangle that it enters, it cannot escape $\mathcal{C}_i$  since, by weak simplicity, it cannot cross the previous edge that it traced, and is thus forced to continue spiraling around $v_i$ indefinitely. This contradicts the assumption that $\gamma$ is closed, finishing the proof. 
\end{proof}

The following geometric lemma will come handy to bound the combinatorics of a closed simple quasigeodesic.

\begin{lemma}\label{lem:quad}
Let $Q$ be a Euclidean quadrilateral consisting of two Euclidean triangles glued along an edge. Then the distance between two opposite sides of $Q$ is lower bounded by the smallest altitude of the two triangles.
\end{lemma}

\begin{proof}
 One easily sees that the distance between two opposite sides of $Q$ is realized by two points $x$ and $y$, one of which, say $x$, can be assumed to be a vertex of $Q$. Now we distinguish two cases, depending on whether the edge $e$ separating the two triangles inside $Q$ is adjacent to $x$ or not. If yes, then the distance between $x$ and $y$ is actually realized by one of the altitudes of the two triangles, and thus the lemma follows. Otherwise the path connecting $x$ to $y$ crosses the edge $e$, and thus the distance between $x$ and $y$ is larger than the distance between $x$ and $e$, and thus bigger than the altitude connecting $x$ to $e$.
\end{proof}

We then have the following proposition showing that some quasigeodesic of bounded combinatorial complexity exists.

\begin{proposition}\label{prop:quasi}
Let $S$ be a polyhedral sphere, let $M$ denote the sum of the edge-lengths of the triangles of $S$, let $h$ denote the smallest altitude of the triangles of $S$, and let $d$ be the maximum degree of a vertex in $S$. Then there exists a weakly simple closed quasigeodesic on $S$ such that the length of $\mathcal{E}(\gamma)$ is bounded by:

  \[\eta_{\gamma}= \lceil(d+1) M/h\rceil\]
\end{proposition}

\begin{proof}

We first observe that $h$ lower-bounds the distance between any vertex and the boundary of its star, and thus in particular the distance between any two vertices. Now, let $\gamma$ be a weakly simple quasigeodesic of length at most $M$, whose existence is guaranteed by Theorem~\ref{thm:existence}. We argue that each subarc of $\gamma$ of length $h$ crosses or follows at most $d+2$ edges, which proves the proposition.

Orient $\gamma$ arbitrarily, and denote by $e=ij$ an arbitrary edge that $\gamma$ crosses or follows. Following $\gamma$ after $e$, the first intersection with the $1$-skeleton of $S$ occurs either at a vertex or an edge. If it occurs at a vertex $v$, then the arc of $\gamma$ between $v$ and the next crossing or vertex has length at least $h$ by the observation, and we are done. If it crosses another edge $f$, without loss of generality, $f$ is also adjacent to $i$. We now look at the connected component of $\gamma$ in $\mathcal{C}_i$ containing the crossings at $e$ and $f$. By Lemma~\ref{lem:nonspiral}, this connected component exits $\mathcal{C}_i$ after at most $d$ crossings. Now, when it exits, it crosses successively two edges adjacent to $i$ and an edge not adjacent to $i$, i.e., two opposite edges of a quadrilateral obtained by gluing two triangles. By Lemma~\ref{lem:quad}, the distance between these two opposite edges is at least $h$. Therefore, a subarc of $\gamma$ of length at most $h$ and starting at an edge crosses or follows at most $d+1$ edges, as required.
\end{proof}

We now have all the tools to prove Theorem~\ref{thm:algorithm}.

\begin{proof}[Proof of Theorem~\ref{thm:algorithm}]
  Let $\gamma$ be a weakly simple closed quasigeodesic whose combinatorial complexity is controlled by $\eta_{\gamma}$ as specified by Proposition~\ref{prop:quasi}. First, we observe that we can assume that this quasigeodesic meets a vertex.To see this, we unfold the sequence of triangles crossed by $\gamma$. Note that this unfolding may a priori have overlapping triangles, as pictured in Figure~\ref{F:test-final} (see also~\cite[Figure~24.20]{demaine}). In this unfolding, we can represent $\gamma$ as a straight line connecting two edges (which are identified to each other in $S$). Now, pushing $\gamma$ in a normal direction does not change the angles at the extremities, and thus preserves the fact that we have a quasigeodesic. Note that it also does not change the length since it in the unfolding, the two connected edges are parallel. Furthermore, while pushing, we do not create self-intersections of $\gamma$ until we reach a vertex: indeed, if there was such a self-intersection outside of a vertex after pushing along a distance $t$, this intersection would be either parallel or transverse. In the latter case, there would have already have been a self-intersection after pushing along a distance $t-\varepsilon$. In the former case, two parallel portions will necessarily reach a vertex, as otherwise they stay parallel in the whole curve, which is impossible for a closed quasigeodesic that does not go through a vertex. Thus, we can do this pushing until we reach a vertex, at which stage the curve $\gamma$ will be weakly simple.

\begin{figure}[h!]
		\begin{center}
		  \includegraphics[width=0.3\textwidth]{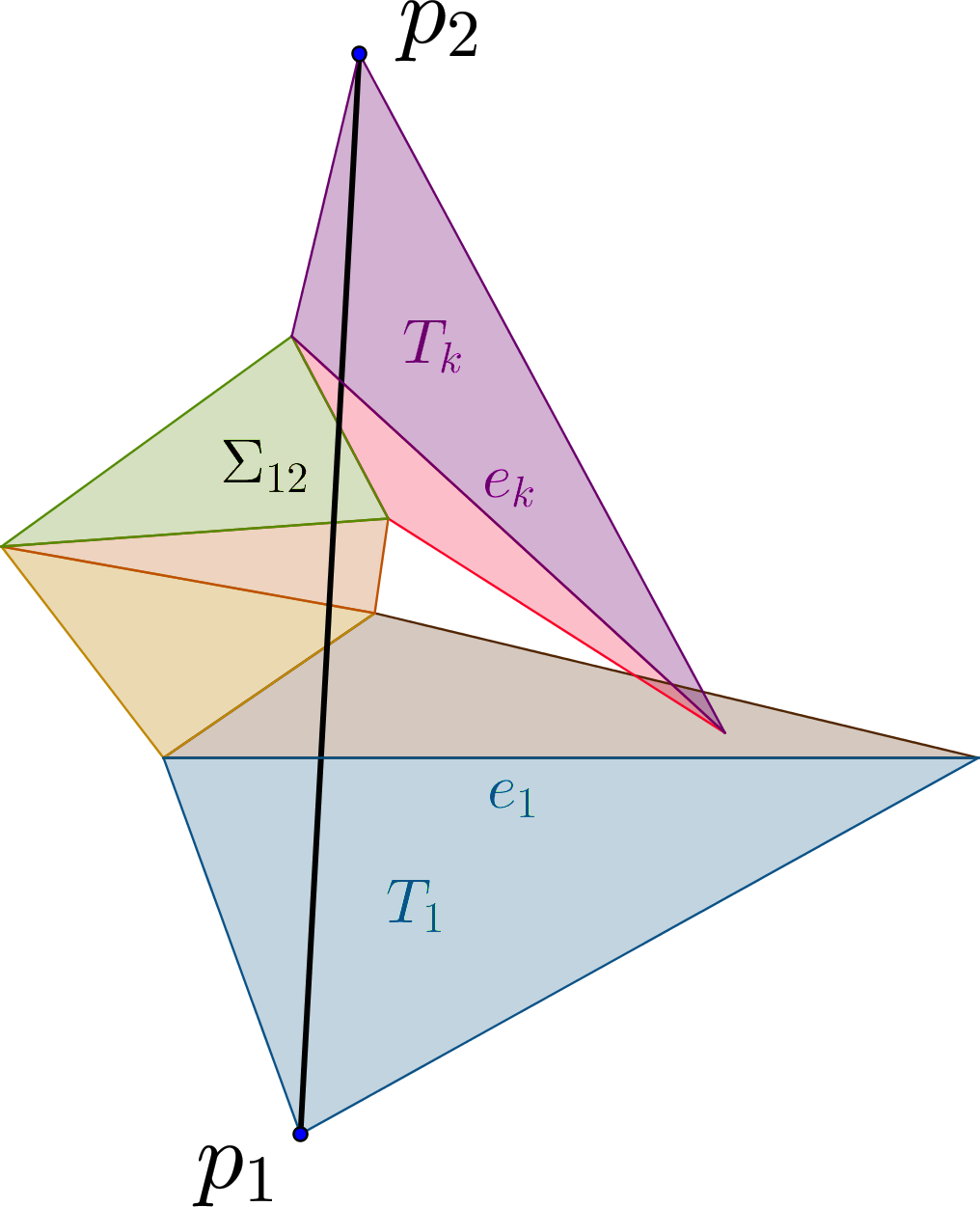}
                  \caption{Unfolding a (tentative) quasigeodesic along the set of edges that it crosses.}
                  \label{F:test-final}
		\end{center}
\end{figure}

  Then, we guess the cyclic word $w$ of size at most $\eta_{\gamma}$ describing the combinatorics of $\gamma$, a weakly simple closed quasigeodesic going through at least one vertex. For each subword $p_1e_1 \ldots e_kp_2$ between two consecutive vertices, if $e_1$ is adjacent to $p_1$, we simply check that the next letter is the other endpoint of $p_1$. Otherwise, we first check that successive letters of that word are adjacent to a common triangle. Then we compute a local unfolding of the polyhedral sphere along the edges $e_1, e_2 \ldots e_k$, i.e., we first place the triangle $T_1$ spanned by $p_1$ and $e_1$, to which we attach along $e_1$ the triangle $T_2$ spanned by $e_2$ and $e_3$, and so on until we reach the last triangle $T_k$ spanned by $e_k$ and $p_2$. Now, in this unfolded picture, we trace the straight line $\Sigma_{12}$ between $p_1$ and $p_2$. There remains to check that the combinatorics of this straight line match those of the guessed word: in the first and last triangles, we check that $\Sigma_{12}$ exits via or follows $e_1$ (or $e_k$), and in each other triangle $T_i$ it suffices to check that the three vertices of $T_i$ are on the sides of $\Sigma_{12}$ prescribed by the edges $e_i$ and $e_{i+1}$ (i.e., if $e_i=ab$ and $e_{i+1}=bc$, then $a$ and $c$ should be one side of $\Sigma_{12}$ while $b$ should be on the other side). Then, we check that the angle between each pair $\Sigma_{i,i+1}$,$\Sigma_{i+1,i+2}$ is within the rules specified by the curvature at the vertex $p_{i+1}$. Finally, we check that this curve is weakly simple, for example via known algorithms~\cite{akitaya,chang2014} or by brute-forcing in exponential time the choice of on which side two overlapping segments can be desingularized.  If all the checks are positive, we have found the unique closed quasigeodesic matching the combinatorics of the word $w$, which is thus weakly simple.
\end{proof}

Finally, let us discuss how to find a simple closed quasigeodesic of bounded length in the case of a convex polyhedron. Following the discussion at the end of Section~\ref{S:existence}, Pogorelov's theorem implies that there exists a simple closed quasigeodesic of length at most $M+\varepsilon$, for an arbitrarily small $\varepsilon$, and allowing as a ``simple'' closed quasigeodesic the degenerate case of a curve connecting twice two vertices of curvature at least $\pi$. This degenerate case is a weakly simple curve that will be found by our algorithm. For the non-degenerate case, the arguments of Proposition~\ref{prop:quasi} apply verbatim to provide a bound on the combinatorics of some simple closed quasigeodesic $\gamma$. If this quasigeodesic goes through at least one vertex, the algorithm described just above finds it, and it is immediate to check that it is simple. If not, we can push it as in the proof of Theorem~\ref{thm:algorithm} to a weakly simple closed geodesic that goes through a vertex, and it will stay simple until it hits that vertex, where it will form an angle exactly $\pi$ in the direction where it came from. Since the total angle at each vertex is at most $2\pi$, this implies that this curve is either degenerate or simple, and in both cases it will be found by our algorithm.

\paragraph*{Acknowledgements.} We thank Francis Lazarus for insightful discussions, and Joseph O'Rourke and the anonymous reviewers for helpful comments.

\bibliographystyle{plainurl}
\bibliography{biblio}

\end{document}